\documentclass[11pt]{article}
\usepackage[a4paper,margin=30mm]{geometry}
\usepackage{amsmath,amssymb,amsthm,mathtools,bm}
\usepackage{microtype}
\usepackage{graphicx}
\numberwithin{equation}{section}
\usepackage{booktabs,array,tabularx}
\usepackage{xcolor}
\usepackage{hyperref}
\usepackage[nameinlink,capitalise]{cleveref}

\hypersetup{pdftitle={Microscopic Parametric Correlations and Spectral Rigidity in the Non-Hermitian Threefold Way},pdfauthor={Yan V. Fyodorov},pdfsubject={Manuscript v3},colorlinks=true,linkcolor=blue!55!black,citecolor=blue!55!black,urlcolor=blue!55!black}
\newcommand{\Tr}{\mbox{Tr}}
\newcommand{\Pf}{\mbox{Pf}}
\newcommand{\E}{\mathbb E}
\newcommand{\C}{\mathbb C}
\newcommand{\R}{\mathbb R}
\newcommand{\1}{\mathbf 1}
\newcommand{\dd}{\,\mathrm d}
\newcommand{\cK}{\mathcal K}
\newcommand{\cO}{\mathcal O}
\newcommand{\cV}{\mathcal V}
\newcommand{\AiD}{AI^\dagger}
\newcommand{\AiiD}{AII^\dagger}

\newcommand{\NLsM}{nonlinear sigma model}
\newcommand{\diag}{\mbox{diag}}
\newcommand{\Cov}{\mbox{Cov}}
\newcommand{\Var}{\mbox{Var}}
\newcommand{\FP}{\mbox{FP}}
\newcommand{\Jrep}{J_{\mathrm r}}
\newcommand{\cZ}{\mathcal Z}
\newcommand{\cI}{\mathcal I}
\newcommand{\cD}{\mathcal D}
\newcommand{\cM}{\mathcal M}
\newcommand{\cQ}{\mathcal Q}

\newtheorem{proposition}{Proposition}[section]

\newtheorem{conjecture}[proposition]{Conjecture}
\theoremstyle{definition}

\newtheorem{remark}[proposition]{Remark}

\title{Microscopic Parametric Correlations and Spectral Rigidity\\
in the Non-Hermitian Threefold Way}
\author{Yan V. Fyodorov\\[2mm]
Department of Mathematics, King's College London,\\
London WC2R 2LS, United Kingdom}
\date{}

\begin{document}
\maketitle

\begin{abstract}
We study microscopic spectral correlations between two nearby parameter
values in the three non-Hermitian Gaussian bulk classes: $A$
(unconstrained), $\AiD$ (complex symmetric), and $\AiiD$ (complex
self-dual, with each degenerate doublet counted once). Parameter dependence
is modelled by stationary matrix Ornstein--Uhlenbeck evolution, whose
microscopic time scale is $N^{-1}$. For every positive integer replica
number $n$, we derive exact finite-$N$ auxiliary-field integral
representations for joint characteristic-polynomial moments and obtain
their bulk asymptotics as $N\to\infty$ with $n$ fixed. With the spectrum
normalized to the unit disk, correlations near a bulk point $z_0$, between
spectral positions $z_1,z_2$ at times $t_1,t_2$, depend only on
$s=N|z_1-z_2|^2+N(1-|z_0|^2)|t_1-t_2|$. Adopting the static
Hermitian/non-Hermitian replica continuation yields explicit two-time
kernels; for $\AiD$ and $\AiiD$ these are replica conjectures and, to our
knowledge, the first microscopic two-time kernels proposed for these
classes. They determine density correlations, cross-time eigenvalue-count
covariances, and equal-time number variances. For a disk of mean eigenvalue
count $y$, the latter obey
$\mbox{Var}\mathcal N_X(D_y)=\kappa_X\sqrt y+
\beta_X/\sqrt y+o(y^{-1/2})$, with
$\kappa_{\AiD}>\kappa_A>\kappa_{\AiiD}$. Exact perturbation theory further
identifies eigenvector nonorthogonality as the mechanism of short-time
spectral diffusion and yields a basis-independent condition number for a
two-dimensional Kramers eigenspace in class $\AiiD$. Direct simulations
agree quantitatively with its conjectured inverse-gamma law.
\end{abstract}

\section{Introduction}

Random-matrix theory entered quantum chaos primarily as a theory of
spectral statistics at a fixed Hamiltonian, most famously through the
Bohigas--Giannoni--Schmit paradigm \cite{BohigasGiannoniSchmit1984}.
Parametric spectral statistics asks the complementary dynamical question.
Given a family of Hermitian Hamiltonians $H(x)$ depending on an external
control $x$, for example a magnetic flux, a boundary deformation, a gate
voltage, or the disorder strength, one studies how its eigenvalues $E_n(x)$ and
associated eigenvectors respond when $x$ is varied.  The infinitesimal observables
are the level velocities $\dot E_n(x)$ and the level curvatures
$\ddot E_n(x)$, whereas finite parameter differences lead to
two-parameter density correlations, spectral cross-form factors, and the
parametric number variance
\cite{GaspardEtAl1990,GoldbergEtAl1991}.  These quantities contain
information which is absent from fixed-parameter level statistics.
Indeed, by the Hellmann--Feynman relation a level velocity probes a
diagonal matrix element of the perturbation, whereas a curvature contains
in addition off-diagonal matrix elements divided by level spacings and is
therefore especially sensitive to avoided crossings.  Universal curvature
distributions and their algebraic tails were obtained in random-matrix
models in \cite{ZakrzewskiDelande1993,vonOppen1994,FyodorovSommers1995},
and velocity statistics were used to probe eigenfunction fluctuations and
localization in disordered systems
\cite{Fyodorov1994,FyodorovMirlin1995}.

A central achievement of the early 1990s was the discovery that, for a
spatially extended perturbation, the full parametric correlations become
universal once the parameter difference is measured in the natural scale
fixed by the mean-square level velocity
\cite{SzaferAltshuler1993,SimonsAltshulerVelocity1993,
SimonsAltshulerUniversality1993}. The same universal functions emerged from several complementary
descriptions.  The \NLsM\ treated the microscopic regime in which levels
move by approximately one mean spacing.  A particularly relevant
Hermitian symplectic application was given by Kravtsov and Zirnbauer, who
used the zero-mode \NLsM\ for the flux-driven Gaussian
symplectic-to-unitary crossover to extract persistent-current statistics
from parametric two-level correlations
\cite{KravtsovZirnbauer1992}.  Since the current carried by a level is
$-\partial E_\nu/\partial\Phi$, this is a level-velocity problem in which
the flux lifts a Kramers degeneracy.  Perturbative diagrammatics described
the overlapping mesoscopic regime, while semiclassical periodic-orbit
theory provided a parallel large-separation description and related the
parametric correlation scale to the underlying classical dynamics
\cite{BerryKeating1994,OzorioLewenkopfMucciolo1998}.

  In the Brownian-motion
formulation, originating with Dyson \cite{Dyson1962,Beenakker1993}, the
matrix parameter becomes a stochastic time and the eigenvalues form an
autonomous interacting diffusion whose Fokker--Planck generator, after a
similarity transformation, becomes an imaginary-time
Calogero--Sutherland Hamiltonian.  Thus in the self-adjoint setting
parametric spectral correlators can be read as correlation functions of
an interacting one-dimensional quantum gas
\cite{SimonsLeeAltshuler1993,BeenakkerRejaei1994,NagaoForrester1998}.
Predictions of this theory were tested in atomic spectra and in
wave-chaotic experiments \cite{SimonsEtAl1993,BertelsenEtAl1999}.  It was
also understood that local or finite-rank perturbations define a
different response problem and need not share the scaling functions of
global perturbations \cite{MarchettiSmolyarenkoSimons2003}.

The non-Hermitian extension is not obtained merely by allowing the level
coordinates to become complex.  Denoting by $z_n(x)$ a simple eigenvalue
with biorthogonally normalized left and right eigenvectors, one has
$\dot z_n=L_n^*\dot X\,R_n$, so that the velocity depends explicitly on
eigenvector nonorthogonality.  Under non-Hermitian matrix Brownian motion
the eigenvalues therefore no longer form an autonomous diffusion: their
quadratic covariations are governed by the left--right overlap matrix
\cite{ChalkerMehlig1998,ChalkerMehlig2000,BourgadeDubach2020,
BourgadeCipolloniHuang2024}. Parametric spectral motion is in this sense an eigenvector observable in
disguise, as already demonstrated in open chaotic scattering through the
relation between perturbation-induced resonance-width shifts and
resonance-state nonorthogonality \cite{FyodorovSavin2012}. 
The complex Ginibre ensemble provides the basic class-$A$ laboratory for
this phenomenon.  A modern account of its static eigenvalue correlations
and eigenvector-overlap statistics is given by Byun and Forrester
\cite{ByunForrester2025}.

The bulk of non-Hermitian random-matrix theory has a threefold
organization.  Besides the symmetry-free complex class $A$, transposition
symmetry gives the complex-symmetric class $\AiD$ and the complex
self-dual class $\AiiD$, whose Hermitian partners have Dyson indices
$\beta=2,\,1,\,4$, respectively.  This non-Hermitian counterpart of the
Wigner--Dyson threefold way was identified in
\cite{HamazakiKawabataKuraUeda2020}.  Two recent breakthroughs have
supplied complementary static descriptions of it.  The replica duality of
Chen, Xiao, Liu, and Ryu gives closed bulk two-point functions by
analytic continuation from the Hermitian Wigner--Dyson classes
\cite{ChenXiaoLiuRyu2026}, whereas the Calogero construction of Xiao,
Chen, Liu, and Ryu gives exact finite-$N$ joint eigenvalue densities as
scattering states of an integrable many-body Hamiltonian
\cite{XiaoChenLiuRyu2026}.  The latter result is especially important
conceptually: outside class $A$ the joint density is not a Coulomb gas
with pairwise interactions.  This exact static Calogero scattering
representation suggests a possible non-Hermitian counterpart of the
Hermitian dynamical mapping, a formulation which would have to
incorporate the eigenvector-overlap degrees of freedom that make
non-Hermitian eigenvalue motion non-autonomous.  The class-$A$
microscopic parametric kernel derived in our earlier work with
Lacroix-A-Chez-Toine \cite{FyodorovLacroix2026} provided a first step in
this direction.  The present work extends that programme to the full
non-Hermitian threefold way, identifying eigenvector nonorthogonality as
the mechanism which governs short-time spectral diffusion.

The complex-symmetric class $\AiD$ has a direct and long-established
physical origin in open quantum scattering.  For a
time-reversal-invariant chaotic system the poles of the scattering matrix
are eigenvalues of an effective complex-symmetric non-Hermitian
Hamiltonian, a structure used by Sommers, Fyodorov, and Titov to describe
analytically the crossover from isolated to overlapping resonances
\cite{SommersFyodorovTitov1999}.  More recently Kurilov and Ostrovsky \cite{KurilovOstrovsky2026}
developed a \NLsM\ approach to the density and the width distribution of
scattering resonances in disordered systems, applicable to an arbitrary
symmetry class and to arbitrary coupling to the external channels.

On the random-matrix side, recent analytical work on the transposition
classes includes characteristic-polynomial identities and dualities
\cite{AkemannAygunKieburgPassler2025,Forrester2025}, spectral-density and
eigenvector-overlap statistics in the complex-symmetric class $\AiD$
\cite{AkemannFyodorovSavin2026}, and fermionic nonlinear sigma models
\cite{KulkarniKawabataRyu2025}.  A complementary study of local spacing
statistics across the full threefold way derives finite-$N$
complex-spacing-ratio formulas in class $A$ and compares numerically the
spacing ratios and the nearest- and next-nearest-neighbour distributions
of all three classes, both in the bulk and at the spectral edge
\cite{AkemannEtAl2026Edge}. Transposition symmetry also arises naturally in open quantum dynamics: it
was realized explicitly in dissipative quantum circuits and incorporated
into the symmetry classification of many-body Lindbladians
\cite{SaRibeiroProsen2021,SaRibeiroProsen2023}.  Its consequences for
open-system spectral statistics, including non-Hermitian Hamiltonians in
classes $\AiD$ and $\AiiD$, are further discussed in
\cite{KawabataEtAl2023}. Complex spectral form factors were introduced and analyzed in the
weakly non-Hermitian crossover by Fyodorov, Khoruzhenko, and Sommers
\cite{FyodorovKhoruzhenkoSommers1997,FyodorovKhoruzhenkoSommers1998}.
Their later use, alongside complex level-spacing statistics, as diagnostics
of dissipative quantum chaos was developed in
\cite{AkemannKieburgMielkeProsen2019,LiProsenChan2021}.

There is at the same time a rigorous dynamical theory on mesoscopic
scales.  In class $A$, Bourgade, Cipolloni, and Huang
\cite{BourgadeCipolloniHuang2024} consider the non-Hermitian matrix
Ornstein--Uhlenbeck process
\[
 \dd X_t=\frac{1}{\sqrt N}\,\dd B_t-\frac12X_t\,\dd t,
\]
where the entries of $B_t$ are independent complex Brownian motions and
the stationary distribution is the complex Ginibre ensemble.  For
$t_1<t_2$ the transition law can equivalently be written as
\[
 X_{t_2}\stackrel{\mathrm d}{=}
 e^{-(t_2-t_1)/2}X_{t_1}+\sqrt{1-e^{-(t_2-t_1)}}\,G,
\]
with $G$ an independent Ginibre matrix.  The correlated pair used in the
present paper is therefore precisely a pair of observations of the
stationary matrix process at two different times, with
$\alpha=e^{-|t_1-t_2|/2}$.

Writing $\lambda_1(t),\ldots,\lambda_N(t)$ for the eigenvalues of $X_t$,
Bourgade \emph{et al.} study the joint fluctuations of smoothed spectral
observables, also called linear statistics,
$L_N(f,t)=\sum_j f(\lambda_j(t))-\E\sum_j f(\lambda_j(t))$, evaluated at
several times.  Their mesoscopic regime uses spectral windows of radius
$N^{-\eta}$ with $0<\eta<1/2$: such a window shrinks on the macroscopic
scale, yet still contains a diverging number of eigenvalues.  They prove
multitime central limit theorems, including for suitable
non-equilibrium independent identically distributed initial data.  Near a
fixed bulk point $z_0$ the limiting covariance depends on spectral and
temporal separation through $|z-w|^2+\chi_0|t_1-t_2|$, where $z,w$ are
spectral positions and $\chi_0=1-|z_0|^2$, so that the natural parabolic
distance is
\[
 d_{\mathrm{par}}\bigl((z,t_1),(w,t_2)\bigr)
 =\bigl(|z-w|^2+\chi_0|t_1-t_2|\bigr)^{1/2},
\]
a time difference thus carrying the scaling dimension of a squared
spectral distance \cite{BourgadeCipolloniHuang2024}.

The same work proves that the limiting space--time field is not
Markovian: its future fluctuations retain information about the past
which is not contained in the field at a single intermediate time.
Bourgade \emph{et al.} trace this memory to the evolving eigenvectors.
More concretely, let $O_{ii}=\|R_i\|^2\|L_i\|^2$ denote the diagonal
left--right eigenvector overlap, with $L_i^*R_i=1$.  After the
corresponding overlap-weighted spectral fields are smeared against smooth
test functions on mesoscopic windows, their limiting two-time covariance
is obtained by integrating the two probes against
$\chi_0^2/(|z-w|^2+\chi_0|t_1-t_2|)^2$.  Since the denominator is
$d_{\mathrm{par}}^4$, this kernel decays algebraically as
$d_{\mathrm{par}}^{-4}$.  Such quartic decay reflects the persistence of correlations mediated by
the evolving eigenvectors.

Their theory concerns smooth linear statistics on scales larger than the
mean eigenvalue spacing.  We instead resolve the complementary regime
$|z-w|\asymp N^{-1/2}$ and $|t_1-t_2|\asymp N^{-1}$.  At these scales a
spectral window contains only $O(1)$ eigenvalues, so smoothing no longer
suppresses the short-time contribution associated with the motion of an
individual eigenvalue.  Our organizing object is a space--time covariance
potential $\cK_X$, which we call the ``eigenvector-memory kernel'', rather
than the static pair correlation itself.  To introduce its time dependence,
consider the stationary matrix Ornstein--Uhlenbeck process at observation
times $t_1,\ldots,t_m$, with $|t_i-t_j|=O(N^{-1})$.  Writing $\lambda_k(t_i)$ for its eigenvalues,
define the microscopic counting density near a bulk point $z_0$ by
\begin{equation}
 \rho_{i,N}(\xi)
 =
 \sum_{k=1}^N
 \delta^{(2)}
 \!\left(\xi-\sqrt N[\lambda_k(t_i)-z_0]\right).
 \label{eq:microscopicden}
\end{equation}
In class $\AiiD$ each degenerate doublet is counted once, and the
limiting mean density in these coordinates is $1/\pi$.  Although
$\rho_{i,N}$ is written as a density, it is a sum of delta functions:
covariance formulas involving it are understood after integration against
smooth probe functions.

For $t_i\ne t_j$, define the dimensionless connected two-time density
correlation $c_{ij}^X$ by
\[
 \E\rho_i(\xi)\rho_j(\eta)
 -
 \E\rho_i(\xi)\E\rho_j(\eta)
 =
 \frac{1}{\pi^2}c_{ij}^X(u),
 \qquad
 u=|\xi-\eta|^2,
\]
so that, denoting by $z$ and $w$ the corresponding points in the original
spectral plane, one has $u=N|z-w|^2$.  At equal times $c_{ii}^X$ denotes
the distinct-eigenvalue part of the connected correlation, the full
covariance containing in addition the self-correlation
$\pi^{-1}\delta^{(2)}(\xi-\eta)$.

Setting now $\chi_0=1-|z_0|^2$, $\tau_{ij}=N|t_i-t_j|$, and
$a_{ij}=\chi_0\tau_{ij}$, the central conjecture of the present paper is
the unified relation
\[
 c_{ij}^X(u)
 =
 \frac{\dd^2}{\dd u^2}
 \left[u^2\cK_X(u+a_{ij})\right],
\]
with the memory kernel $\cK_X$ obtained explicitly in closed form for all
three symmetry classes in the generic spectral bulk.  All dependence on
the symmetry class is thus carried by $\cK_X$, whereas spectral
separation and time enter universally through the single combination
$u+a_{ij}$.  At equal times the formula gives the static distinct-level
correlation, while at distinct times it describes the progressive loss of
spectral memory.  Because the Ornstein--Uhlenbeck process is stationary,
applying the formula to every pair $(t_i,t_j)$ produces the full
covariance matrix of density observables at any finite collection of
times; integrating the same space--time density covariance over spectral
windows gives in turn the pairwise multitime covariance matrix of
eigenvalue counts.  Differentiation of $\cK_X$ therefore produces
microscopic density correlations, whereas integration produces counting
statistics.  Here ``multitime'' means this complete second-order
covariance structure: the present calculation determines only the two-point intensity, not joint
intensities at three or more spectral positions.

This relation has a nontrivial consequence.  Setting
$F_X(u)=u^2\cK_X(u)$, the static distinct-level correlation obtained in
\cite{ChenXiaoLiuRyu2026} and reproduced below in \eqref{eq:static} fixes
$F_X''(u)$ but not $F_X(u)$ itself, and therefore leaves the two affine
contributions $b_0+b_1u$ undetermined, which in terms of $\cK_X$ are the
invisible modes $b_0/u^2+b_1/u$.  Short-distance regularity, or
equivalently the absence of an additional $u^{-2}$ singularity
incompatible with a finite covariance measure, rules out the first mode.
The second is fixed by the diagonal part of the equal-time covariance:
every eigenvalue is correlated with itself, giving a planar delta mass
whose coefficient equals the local mean density, namely $1/\pi$ in our
bulk normalization.  This self-correlation is called the ``contact term''
in field-theory language and the diagonal part of the second-moment
measure in point-process language.

At a positive time difference the contact mode broadens into a unit-mass
short-distance component of the endpoint cross-covariance.  Under the
shift $u\mapsto u+a$ the mode $1/u$ produces $2a^2/(u+a)^3$, which
converges back to the delta mass as $a\downarrow0$.  In the further
scaling $u=ax$, $a\downarrow0$, eigenvalue perturbation theory identifies
this profile with the displacement of a locally followed eigenvalue in
the classes for which the overlap law is known.  At fixed positive $a$, however, the endpoint spectra alone do not define a
pathwise matching of their eigenvalues.  The complementary sum rule follows
instead from the requirement that the connected correlation vanish at large
spectral separation.

For class $A$ the microscopic parametric density correlation and the
associated disk-count covariance were derived in our earlier joint work
with Lacroix-A-Chez-Toine by the replica $Q$-matrix method
\cite{FyodorovLacroix2026}.  In a companion manuscript in preparation
with Cuppone \cite{CupponeFyodorov2026} we obtain an independent
Kac--Rice derivation of the same class-$A$ spectral formula.  That work
starts from an exact finite-$N$ Kac--Rice representation for two
eigenvalues and their normalized right eigenvectors, deriving the
corresponding right-eigenvector-resolved microscopic bulk density by
controlled saddle asymptotics and recovering the parametric spectral
two-point function of \cite{FyodorovLacroix2026} after integration over
the right-eigenvector overlap.

Accordingly, class $A$ serves here as a calibration rather than as a new
result.  The new contributions are the following:
(i) the exact finite-$N$ correlated auxiliary-field integrals for $\AiD$
and $\AiiD$,
(ii) their large-$N$ saddle analysis at fixed positive integer $n$,
showing that the parameter separation enters through the additive shift
$u\mapsto u+a$,
(iii) the resulting replica conjectures for the two-time correlations,
(iv) a unified counting formalism with explicit subleading rigidity
coefficients, and
(v) the intrinsic Kramers-plane diffusivity.
To our knowledge the proposed
$\AiD$ and $\AiiD$ kernels are the first microscopic two-time kernels in
these classes.  The genuinely parametric input is the covariance
deformation of the exact finite-$N$ integral, whose large-$N$ saddle
contribution produces the combined variable $u+a$.

\subsection{A guide to logical status}

Several logically distinct steps enter the calculation.  Throughout the
paper we call a spectral statement a \emph{conjecture} when it requires
analytically continuing the fixed-positive-integer-$n$ formulas to $n=0$
and interchanging this continuation with the large-$N$ limit.  The exact
finite-$N$ auxiliary-field identities and the saddle-point reduction at
fixed positive integer $n$ are not conjectural.  The following table
summarizes the status of the principal results.

\begin{center}
\small
\begingroup
\renewcommand{\arraystretch}{1.08}
\begin{tabularx}{\textwidth}{
  >{\raggedright\arraybackslash}p{42mm}
  >{\raggedright\arraybackslash}p{31mm}
  X}
\toprule
object & status & meaning\\
\midrule
Covariance tensors, Grassmann averages, auxiliary-field identities, and
Kramers-plane perturbation theory
&
exact at finite $N$
&
no replica continuation is involved\\
\addlinespace[4pt]

Compact-group reduction
&
fixed positive integer $n$
&
a finite-dimensional Laplace asymptotic, uniform in the stated bulk regime\\
\addlinespace[4pt]

Two-time $\AiD$ and $\AiiD$ density kernels
&
replica conjecture
&
requires analytic continuation to $n=0$ and interchange with
$N\to\infty$\\\addlinespace[4pt]

Abel formulas and rigidity coefficients
&
conditional corollaries
&
exact calculations from the displayed kernels; their interpretation as
spectral limits inherits the replica status\\
\addlinespace[4pt]

Eigenvector-overlap statistics: two-eigenvalue correlations and the
$\AiiD$ Kramers-plane overlap distribution
&
exact in class $A$; conjectural for $\AiD$ and $\AiiD$
&
would require an exact finite-$N$ calculation with explicit
eigenvector-overlap insertions, or an independent eigenvector derivation\\
\bottomrule
\end{tabularx}
\endgroup
\end{center}

Thus exact identities, controlled fixed-$n$ asymptotics, and replica
conjectures are not used interchangeably.  The distinction is
mathematical, but it is also physically useful: it identifies precisely
the step in the $\sigma$-model calculation which still lacks a rigorous
mathematical justification.

\subsection{Two dictionaries for the same problem}

It is convenient to pause and translate the principal objects.  In
field-theory language one differentiates a replicated generating function
twice, once for each observation time; equivalently, one asks for the
joint density of finding an eigenvalue near each of two prescribed points
and subtracts the product of the one-point densities.  In point-process
language this is called a two-time cross-intensity, or its connected
covariance measure.  The radial density is $c_a^X$, whereas $\cK_X$ is
its twice-integrated potential, and integration over two windows gives
the covariance of their eigenvalue counts.  Thus
\[
 \begin{array}{c|c}
 \text{point-process language}&\text{field-theory language}\\ \hline
 \text{two-time cross-intensity}&\text{two density insertions}\\
 \text{contact-mode component}&\text{diagonal/self contraction}\\
 \text{covariance potential }\cK_X&\text{soft-mode propagator with sources}\\
 \text{large-window boundary term}&\text{infrared memory contribution}
 \end{array}
\]
The covariance and disk-overlap identities implement this dictionary
literally.  The identification of the spectral potential with an
overlap-weighted two-eigenvalue correlation is independently established
in class $A$, whereas its extension to the other two classes is discussed
separately in \cref{sec:overlaps}.

There is a second dictionary concerning eigenvectors.  In numerical
linear algebra $O_{ii}$ is the square of an eigenvalue condition number,
whereas in quantum-mechanical and wave scattering by open systems it
appears as the ``Petermann nonorthogonality factor'', see e.g. \cite{PatraSchomerusBeenakker2000,SchomerusFrahmPatraBeenakker2000}. 
 Its stochastic meaning is particularly direct: for the class-$A$ Ornstein--Uhlenbeck
process the eigenvalue martingale satisfies
$\dd\langle\lambda_i,\overline{\lambda_i}\rangle_t=O_{ii}(t)\,\dd t/N$, so
that $O_{ii}/N$ is the instantaneous diffusion coefficient of the
eigenvalue in the complex plane
\cite{BourgadeDubach2020,BourgadeCipolloniHuang2024}.  The
complex-symmetric class $\AiD$ carries the corresponding factor
$2O_{ii}/N$, whereas in the self-dual class $\AiiD$ the role of the
scalar overlap is played by the basis-independent Kramers-plane invariant
derived below.  We refer to the appropriately normalized overlap either
as a dimensionless eigenvector overlap or, when emphasizing its dynamical
role, as a normalized eigenvalue diffusivity, the overlap terminology
itself originating in \cite{ChalkerMehlig1998,ChalkerMehlig2000}.  These
names emphasize different uses of the same biorthogonal invariant.

For a normal matrix the left and right eigenvectors may be chosen
identical and orthonormal, so that the diagonal overlap is $O_{ii}=1$.
For the complex Ginibre ensemble the diagonal bulk overlap $O_{ii}$ is
typically of order $N$ \cite{ChalkerMehlig1998,ChalkerMehlig2000}, and
after normalization by its position-dependent bulk scale it converges to
a heavy-tailed inverse-gamma distribution
\cite{BourgadeDubach2020,Fyodorov2018}.  This is why a matrix
perturbation of entry size $N^{-1}$ already moves an eigenvalue by one
microscopic spacing.  The scale $|t_1-t_2|\asymp N^{-1}$ is therefore not
guessed from eigenvalue repulsion: it follows from eigenvector
amplification.

\subsection{What is genuinely parametric here?}

A static two-point correlation compares two distinct eigenvalues of the
same matrix, whereas a parametric correlation compares the spectra of two
correlated matrices.  At equal parameter values the spectral-density
covariance contains a diagonal delta-function contribution,
conventionally removed from $R_2$, together with the distinct-level
correlation.  At a positive parameter separation this contact
contribution broadens into a finite-width component and must be retained
in the full two-time density correlation.  In the short-time limit
perturbation theory identifies its leading peak with the displacement of
a locally followed eigenvalue, and it is in this asymptotic sense that we
call it the ``self branch''.

For a fixed positive parameter separation, however, assigning the
endpoint eigenvalues to individual trajectories requires knowledge of the
matrix evolution at all intermediate times together with a prescription
for continuously tracking its eigenvalues.  The two-matrix intensity
depends only on the endpoint spectra and does not determine such a
pathwise assignment.

This distinction explains why the parametric correlation cannot be
reconstructed uniquely from the known static pair function $R_2^X(u)$
obtained in \cite{ChenXiaoLiuRyu2026} by the formal substitution
$u\mapsto u+a$.  The static answer determines a second derivative, and
contact normalization fixes the coefficient of the singular mode.  The
velocity calculation in \cref{sec:self} independently reproduces this
normalization, identifying $a$ as the natural mean-square scale of
short-time eigenvalue motion.  The singular mode is thus the dynamical
continuation of the equal-time self-correlation, broadened by
nonorthogonality-enhanced eigenvalue diffusion.

\section{The three Gaussian ensembles and microscopic variables}

\subsection{Ensembles}

We normalize all limiting spectra to the unit disk. In class $A$ \cite{ByunForrester2025},
$X$ is an $N\times N$ complex Ginibre matrix with $\E X_{ij}=0$ and
$\E X_{ij}\overline{X_{kl}}=N^{-1}\delta_{ik}\delta_{jl}$.
In class $\AiD$, $X=X^T$ is complex symmetric.  Set
$X=(G+G^T)/\sqrt2$, with $G$ a class-$A$ matrix.  Then
\begin{equation}
 \E X_{ij}\overline{X_{kl}}
 =\frac1N(\delta_{ik}\delta_{jl}+\delta_{il}\delta_{jk}),
 \qquad \E X_{ij}X_{kl}=0.
 \label{eq:AIcov}
\end{equation}
Thus an off-diagonal entry has variance $1/N$ and a diagonal entry has
variance $2/N$.  No further global rescaling is intended.  The second term, $\delta_{il}\delta_{jk}$, is the permuted-index contraction.
For the bilinear governing the first-order motion of a complex-symmetric
eigenvalue, it contributes the same amount as the direct contraction
$\delta_{ik}\delta_{jl}$ and therefore doubles the conditional eigenvalue
velocity variance.

For class $\AiiD$, matrices act on $\C^{2N}$.  Let
\begin{equation}
 \mathbb J=\1_N\otimes\begin{pmatrix}0&1\\-1&0\end{pmatrix},
 \qquad \mathbb J^T=-\mathbb J,\qquad \mathbb J^2=-\1_{2N}.
 \label{eq:physicalJ}
\end{equation}
Let $\{\xi_{ij},\eta_{ij}\}_{i,j=1}^{2N}$ be independent standard real
Gaussian variables and set
$G_{ij}=(\xi_{ij}+i\eta_{ij})/\sqrt{4N}$.  Thus the entries of $G$ are
independent and satisfy
$\E G_{ij}=0$, $\E|G_{ij}|^2=1/(2N)$, and $\E G_{ij}^2=0$. Set
\begin{equation}
 X=\frac{G+\mathbb JG^T\mathbb J^{-1}}{\sqrt2},\qquad
 X=\mathbb JX^T\mathbb J^{-1}.
 \label{eq:selfdual}
\end{equation}
A direct contraction gives the complete tensor
\begin{equation}
 \E X_{ij}\overline{X_{kl}}
 =\frac1{2N}(\delta_{ik}\delta_{jl}+\mathbb J_{il}\mathbb J_{jk}),
 \qquad \E X_{ij}X_{kl}=0.
 \label{eq:AIIcov}
\end{equation}
Equivalently, $Y=X\mathbb J$ is complex antisymmetric and its independent
entries $Y_{ij}$, $i<j$, have variance $1/(2N)$.  This equivalence is useful
for the Pfaffian generating function below.

Almost surely each distinct self-dual eigenvalue has multiplicity two and a
two-dimensional eigenspace.  Indeed, the discriminant of the degree-$N$
Pfaffian polynomial is nonzero for a diagonal matrix with distinct doublet
values, so its vanishing is a Gaussian null event. Let
\[
 V_\lambda
 =
 \bigcup_{k\geq1}\ker\!\left[(X-\lambda\1_{2N})^k\right]
 =
 \ker\!\left[(X-\lambda\1_{2N})^{2N}\right]
\]
be the generalized eigenspace associated with $\lambda$. Self-duality implies that $V_\lambda$ and $V_\mu$ are symplectically
orthogonal whenever $\lambda\ne\mu$, and since the ambient symplectic form
is nondegenerate, its restriction to each $V_\lambda$ is nondegenerate as
well.  On the event that the reduced
Pfaffian polynomial has simple zeros, $\dim V_\lambda=2$.  In a symplectic
basis of this two-dimensional space, the restriction
$A=X|_{V_\lambda}$ satisfies $A^TJ_2=J_2A$, which forces
$A=\lambda\1_2$.  Hence $V_\lambda=\ker(X-\lambda\1)$: every eigenvalue is
a semisimple doublet. Every spectral sum here counts that doublet
once.  This
transposition-self-dual $\AiiD$ ensemble, defined by
$X=\mathbb JX^T\mathbb J^{-1}$, should not be confused with the
quaternion-real, or symplectic, Ginibre ensemble, defined by
$G=\mathbb J\overline G\mathbb J^{-1}$.  In the latter ensemble, generic
non-real eigenvalues occur in complex-conjugate pairs rather than as
degenerate doublets; its eigenvector angles and left--right overlaps were
studied by Dubach \cite{Dubach2021}.

 The limiting support of all three constructions above is the unit disk,
with $N$ distinct points.  For reference their Gaussian weights
are $e^{-N\Tr XX^\dagger}$, $e^{-(N/2)\Tr XX^\dagger}$, and
$e^{-N\Tr XX^\dagger}$ on their respective matrix spaces.

\subsection{Correlated copies}

Fix one of the three transposition-symmetry classes
$A,\AiD,\AiiD$, and let $X_1$ and $X_2$ be independent matrices drawn
from its Gaussian ensemble, with the normalization specified above. Set
\begin{equation}
 X_\alpha=\alpha X_1+\sqrt{1-\alpha^2}\,X_2,
 \qquad
 \alpha=1-\frac{v^2}{2N}+O(N^{-2}).
 \label{eq:interpolation}
\end{equation}
We observe the spectra near a fixed bulk point $z_0$, $|z_0|<1$, and write
\[
 \chi_0=1-|z_0|^2,
 \qquad
 z_1=z_0+\frac{\omega}{2\sqrt N},
 \qquad
 z_2=z_0-\frac{\omega}{2\sqrt N},
\]
\begin{equation}
 u=|\omega|^2=N|z_1-z_2|^2,
 \qquad
 a=v^2\chi_0,
 \qquad
 s=u+a.
 \label{eq:scales}
\end{equation}
For stationary Ornstein--Uhlenbeck evolution, $\alpha=e^{-|t_1-t_2|/2}$.
Consequently $|t_1-t_2|=\tau/N$ gives $a=\chi_0\tau+o(1)$.  Thus $u+a$ is exactly the
microscopic version of the parabolic distance found at mesoscopic scales.

For clarity, we represent the microscopic spectrum by its rescaled empirical
counting density (cf. \eqref{eq:microscopicden})
\begin{equation}
 \rho_{\sigma,N}(\xi)=\sum_{j=1}^N
 \delta^{(2)}\!\left(\xi-\sqrt N[\lambda_j(X_\sigma)-z_0]\right).
 \label{eq:microscopicfield}
\end{equation}
This is a random distribution, equivalently the density of the rescaled
eigenvalue counting measure.  Each distinct eigenvalue carries unit mass, and in class $\AiiD$ each
doubly degenerate eigenvalue is counted once.
Each sum has $N$ distinct points.  The density notation is physically
convenient although it contains delta functions.  Precisely, one first
smears it with a smooth probe $\varphi$,
\[
 \langle\rho_{\sigma,N},\varphi\rangle
 =\sum_j\varphi\!\left(\sqrt N[\lambda_j(X_\sigma)-z_0]\right),
\]
and the limiting covariance formulas below mean equality after smearing in
both variables.  This short test-function formulation is only a precise way
of reading the familiar density--density correlator.  Approximating the disk indicator from inside and outside by smooth
functions gives the corresponding disk-count covariance.

Write $\rho_0$ and $\rho_a$ for the limiting bulk-scaled counting densities
at the two parameter values, each having mean intensity $1/\pi$.  For $a>0$,
define $c_a^X$ by
\begin{equation}
 \E\rho_0(\omega/2)\rho_a(-\omega/2)
 -\E\rho_0(\omega/2)\E\rho_a(-\omega/2)
 =\frac1{\pi^2}c_a^X(u).
 \label{eq:defc}
\end{equation}
At $a=0$ and away from the coincident point $\omega=0$, we denote the
distinct-level connected correlation by $c_0^X(u)$, where
$u=|\omega|^2>0$. At equal parameters, the full static covariance, including the
coincident-point self-correlation, is the distribution
\begin{equation}
 C_0^X(\omega)
 =
 \frac1\pi\delta^{(2)}(\omega)
 +\frac1{\pi^2}c_0^X(|\omega|^2).
 \label{eq:staticmeasure}
\end{equation}
The first term is the contact contribution supported at $\omega=0$.
The distinct-level pair intensity $R_2^X$ satisfies
$\pi^2R_2^X(u)=1+c_0^X(u)\to1$.  This convention prevents a contact mass from
being lost when the positive-time formula is taken to zero time.

\subsection{Notation used throughout}

\begin{center}
\begin{tabularx}{\textwidth}{>{\raggedright\arraybackslash}p{28mm}X}
\toprule
symbol & meaning\\
\midrule
$X\in\{A,\AiD,\AiiD\}$ & one of the three transposition-symmetry bulk classes\\
$\beta=2,1,4$ & Dyson index of the Hermitian partner\\
$\chi_0=1-|z_0|^2$ & local bulk eigenvector-overlap and
eigenvalue-diffusivity factor\\
$u=N|z_1-z_2|^2$ & squared microscopic spatial separation\\
$a=v^2\chi_0$ & intrinsic microscopic parameter or time separation\\
$\cK_X$ & eigenvector-memory/eigenvalue-diffusivity kernel\\
$c_a^X$ & connected two-time microscopic density kernel\\
$K_X(a,y)$ & covariance of eigenvalue counts in a disk with mean count $y$\\
$\Sigma_X^2(y)$ & equal-time number variance\\
$\cV_X(a,y)$ & variance of the difference of the two counts\\
$\cO_K$ & basis-independent analogue of the squared eigenvalue condition number for
a Kramers doublet
\\ \\
\bottomrule
\end{tabularx}
\end{center}

\section{Main results: one kernel, several observables}

\subsection{The three kernels}

Define
\begin{equation}
 \mathfrak A(s)=\sinh^2\!\frac{s}{2}\,\frac{\dd}{\dd s}
 \left[\frac{sE_1(s/2)}{\sinh(s/2)}\right],
 \qquad
 E_1(x)=\int_x^\infty\frac{e^{-t}}t\dd t,
 \label{eq:Adef}
\end{equation}
and
\begin{equation}
 \mathfrak B(s)=\frac12e^{-2s}\frac{\dd}{\dd s}
 \left[se^s\operatorname{Shi}(s)\right],
 \qquad
 \operatorname{Shi}(s)=\int_0^s\frac{\sinh t}{t}\dd t.
 \label{eq:Bdef}
\end{equation}
The central objects of the present paper are the three memory kernels,
defined for $s>0$ by
\begin{equation}
\boxed{
\begin{aligned}
 \cK_A(s)&=\frac{1-e^{-s}}{s^2},\\
 \cK_{\AiD}(s)&=-\frac{2\mathfrak A(s)}{s^2},\\
 \cK_{\AiiD}(s)&=\frac{\mathfrak B(s)}{s^2}
 =\frac{e^{-s}}{2s^2}\bigl[(1+s)\operatorname{Shi}(s)+\sinh s\bigr].
\end{aligned}}
\label{eq:threeK}
\end{equation}
For quick comparison, their behaviour at the two limits is listed below,
where $\gamma$ denotes Euler's constant.
\begin{center}
\begin{tabular}{c c c}
\toprule
class $X$ & $s\downarrow0$ & $s\to\infty$\\
\midrule
$A$ & $s^{-1}-\tfrac12+\tfrac16s+O(s^2)$
& $s^{-2}+O(e^{-s}s^{-2})$\\
$\AiD$ & $s^{-1}-\tfrac12-\tfrac{s}{12}(\log s+\gamma-\log2-2)+\cdots$
& $2s^{-2}-4s^{-3}+12s^{-4}+\cdots$\\
$\AiiD$ & $s^{-1}-\tfrac12+\tfrac19s-\tfrac{2}{225}s^3+\cdots$
& $\tfrac12s^{-2}+\tfrac12s^{-3}+\tfrac34s^{-4}+\cdots$\\
\bottomrule
\end{tabular}
\end{center}

The common short-distance singularity is dynamical and local, whereas the
long-range coefficient remembers the Hermitian Dyson index:
\begin{equation}
 \cK_\beta(s)\sim\frac{2}{\beta s^2},
 \qquad \beta=1,2,4.
 \label{eq:infrared}
\end{equation}

\subsection{Density and counting transforms}

\begin{conjecture}[Microscopic two-time density correlation]
Use the analytic continuation of the compact-group integrals specified in
\cref{sec:replica}, with the same branch as in the static threefold-way
calculation.  For a positive parameter separation $a>0$, the predicted
dimensionless connected two-time density-correlation function, normalized
as in \eqref{eq:defc}, is
\begin{equation}
 \boxed{c_a^X(u)=\frac{\dd^2}{\dd u^2}
 \left[u^2\cK_X(u+a)\right].}
 \label{eq:masterdensity}
\end{equation}
\end{conjecture}
At equal parameters the same continuation reproduces the independently known
static result
\begin{equation}
 \pi^2R_2^X(u)=1+\frac{\dd^2}{\dd u^2}
 \left[u^2\cK_X(u)\right].
 \label{eq:static}
\end{equation}
Substituting the three explicit kernels in \eqref{eq:threeK} into the
static relation \eqref{eq:static} reproduces the closed bulk
pair-correlation functions obtained in \cite{ChenXiaoLiuRyu2026}. 
 The genuinely parametric input is derived in
\cref{sec:sigmamodel}: the exact correlated finite-$N$ auxiliary-field
identities are stated in \cref{prop:HS}.  The change in the auxiliary-field action generated by the matrix
correlation coefficient $\alpha$, defined by
$X_\alpha=\alpha X_1+\sqrt{1-\alpha^2}\,X_2$ in
\eqref{eq:interpolation}, is computed in \eqref{eq:masscalculation}.
  In the large-$N$ saddle-point analysis at fixed positive integer $n$,
this additional term combines with the microscopic spectral separation
$u$ to produce the single variable $u+a$; see
\cref{prop:fixedreplica}. These results establish the shift $u\mapsto u+a$ before replica
continuation: it cannot be inferred from agreement with the static formula
at $a=0$ alone.

For a disk $D$ in the bulk-scaled spectral plane with mean eigenvalue count
$y$, let $\mathcal N_0(D)$ and $\mathcal N_a(D)$ denote the numbers of eigenvalues of the two
matrices lying in $D$.  In class $\AiiD$, each degenerate doublet is counted
once.  Assuming the replica-predicted density formula
\eqref{eq:masterdensity}, the geometric disk-overlap identity
\eqref{eq:lenscovariance}, derived below, gives
\begin{equation}
 K_X(a,y):=\operatorname{Cov}(\mathcal N_0(D),\mathcal N_a(D))
 =\frac{y}{\pi}\int_0^{4y}
 \frac{\sqrt u\,\cK_X(u+a)}{\sqrt{4y-u}}\dd u.
 \label{eq:countcov}
\end{equation}
In particular,
\begin{equation}
 \Sigma_X^2(y):=\operatorname{Var}\mathcal N(D)
 =\frac{y}{\pi}\int_0^{4y}
 \frac{\sqrt u\,\cK_X(u)}{\sqrt{4y-u}}\dd u.
 \label{eq:numbervariance}
\end{equation}
The number-difference variance is
\begin{equation}
 \cV_X(a,y)=\operatorname{Var}[\mathcal N_a(D)-\mathcal N_0(D)]
 =2[\Sigma_X^2(y)-K_X(a,y)].
 \label{eq:diffvariance}
\end{equation}
Define $\cK_X^{\rm reg}(s)=\cK_X(s)-1/s$; it has the finite limit $-1/2$
at the origin, although its derivatives need not remain bounded there.
The contribution of $1/s$ can be integrated explicitly:
\begin{equation}
 \cV_X(a,y)=2y\sqrt{\frac{a}{a+4y}}
 +\frac{2y}{\pi}\int_0^{4y}
 \frac{\sqrt u[\cK_X^{\rm reg}(u)-\cK_X^{\rm reg}(u+a)]}
 {\sqrt{4y-u}}\dd u.
 \label{eq:squareroot}
\end{equation}
For fixed $y>0$, the leading response as $a\downarrow0$ is
$\cV_X(a,y)=\sqrt{ay}+O(a)$.  It is the counting-statistics image of the
small-parameter self peak. Equation \eqref{eq:squareroot} is an exact identity
between the proposed kernel transforms for every $a>0$ and separates the universal square-root response generated by the common
$1/s$ singularity from the symmetry-class-dependent contribution of the
regular part of the kernel. In particular, the first term produces the nonanalytic behavior
$\sqrt{ay}$ as $a\downarrow0$, while the second contributes only at order
$O(a)$.

\subsection{Rigidity coefficients}

All spectral-limit statements in this subsection are conditional on the
replica conjecture \eqref{eq:masterdensity}, whereas the integrations and
asymptotic evaluations of the displayed kernels are exact.
As $y\to\infty$, \eqref{eq:numbervariance} yields the perimeter law
$\Sigma_X^2(y)=\kappa_X\sqrt y+o(\sqrt y)$,
where
\begin{equation}
 \kappa_X=\frac1{2\pi}\int_0^\infty
 \sqrt u\,\cK_X(u)\dd u.
 \label{eq:kappa}
\end{equation}
The three values are
\begin{equation}
\begin{aligned}
 \kappa_{\AiD}
 &=\frac1{\sqrt\pi}\left[3-
 \frac{\pi+2\log(1+\sqrt2)}{2\sqrt2}\right]
 =0.7142944933\ldots,\\
 \kappa_A&=\frac1{\sqrt\pi}=0.5641895835\ldots,\\
 \kappa_{\AiiD}
 &=\frac{3\sqrt2-\log(1+\sqrt2)}{4\sqrt\pi}
 =0.4740979713\ldots.
\end{aligned}
\label{eq:kappas}
\end{equation}
This gives a precise sense in which static spectral rigidity is an integrated
record of eigenvector memory.

\section{\texorpdfstring{Class $A$}{Class A} as a normalization laboratory}
\label{sec:classA}

Class $A$ provides the best-calibrated benchmark.  Its equal-parameter
spectral correlations are determinantal \cite{ByunForrester2025}, its
diagonal eigenvector-overlap law is known independently
\cite{BourgadeDubach2020,Fyodorov2018}, and its parametric spectral
correlation has both a replica derivation and an independent
eigenvector-resolved Kac--Rice derivation in a companion manuscript
\cite{FyodorovLacroix2026,CupponeFyodorov2026}.  We therefore use class $A$
to fix and test all normalization conventions before turning to the
complex-symmetric class $\AiD$ and the complex self-dual class $\AiiD$.  Substitution of
$\cK_A$ into \eqref{eq:masterdensity} gives
\begin{equation}
 c_a^A(u)=\frac{\dd^2}{\dd u^2}
 \left[\frac{u^2}{(u+a)^2}(1-e^{-(u+a)})\right].
 \label{eq:classAcompact}
\end{equation}
At $a=0$ this reduces to $-e^{-u}$, and hence
$\pi^2R_2^A(u)=1-e^{-u}$,
the Ginibre bulk kernel.  This check simultaneously fixes the convention
$u=N|z_1-z_2|^2$ and the density normalization $1/\pi$.

The same formula isolates the contact-mode component,
\[
 \cK_A(s)=\frac1s-\frac{s-1+e^{-s}}{s^2},
\]
the first term of which gives $2a^2/(u+a)^3$, whereas the remainder has
total radial weight $-1$, cancelling its unit mass for $a>0$.  This is an
exact algebraic decomposition of the proposed connected kernel.  In the
additional limit $a\downarrow0$, $u=ax$, the first term is the leading
self contribution, whereas at finite $a$ the decomposition does not by
itself resolve individual eigenvalue trajectories.

There is also an eigenvector check.  Conditional on $X$, the real and
imaginary parts of the first-order eigenvalue velocity
$\dot\lambda_i=\ell_i^\dagger Yr_i$ are independent centered real Gaussian
variables with equal variance.  The corresponding displacement satisfies
$\delta\lambda_i=\varepsilon\dot\lambda_i+O(\varepsilon^2)$ with a Gaussian
leading term.  At a bulk point $z_0$, where
$\chi_0=1-|z_0|^2$, the complex Ginibre overlap $O_{ii}/(N\chi_0)$
converges to the inverse-gamma law with density
$p(m)=m^{-3}e^{-1/m}\1_{m>0}$, as obtained in
\cite{BourgadeDubach2020,Fyodorov2018}.
Integrating the conditional Gaussian density over $m$ gives
\[
 \int_0^\infty \frac{e^{-|p|^2/(am)}}{\pi am}\,
 m^{-3}e^{-1/m}\dd m
 =\frac{2a^2}{\pi(a+|p|^2)^3}.
\]
After removing the planar factor $\pi^{-1}$ and writing $u=|p|^2$, one
obtains the same rational expression as in \eqref{eq:selfpeak}.  Here $p$ is the
first-order displacement: comparison with the full spectral response concerns
the scaling $a\downarrow0$, $|p|^2/a$ fixed.  In that regime the independently
known overlap law checks the normalization of $a$.

For counts, \eqref{eq:numbervariance} becomes
\[
 \Sigma_A^2(y)=\frac{y}{\pi}\int_0^{4y}
 \frac{1-e^{-u}}{u^{3/2}\sqrt{4y-u}}\dd u.
\]
Taking the large-window limit of this integral gives the class-$A$
rigidity coefficient $\kappa_A$, namely the coefficient of the leading
$\sqrt y$ growth of the disk number variance:
\[
 \kappa_A
 =
 \frac1{2\pi}\int_0^\infty\frac{1-e^{-u}}{u^{3/2}}\dd u
 =
 \frac1{\sqrt\pi}.
\]
Every formula in the later classes has been calibrated against these three
tests: the static pair function, the same-level velocity mixture, and the
large-window count variance.

\section{Exact correlated integrals and the parameter mass}
\label{sec:sigmamodel}

\subsection{The generating function before any \texorpdfstring{large-$N$}{large-N} limit}

Write $X_+=X_1$ and $X_-=X_\alpha$.  Define the polynomial that counts each
distinct eigenvalue once by
\begin{equation}
 P_X(z;X)=\begin{cases}
 \det(z-X),&X=A,\AiD,\\
 \Pf[(z-X)\mathbb J]/\Pf(\mathbb J),&X=\AiiD.
 \end{cases}
 \label{eq:countingpolynomial}
\end{equation}
It is monic of degree $N$ in every class.  For an integer $n\ge1$ set
\begin{equation}
 \cZ_{N,n}^X(z_+,z_-;\alpha)
 =\E\prod_{\sigma=\pm}|P_X(z_\sigma;X_\sigma)|^{2n}.
 \label{eq:finitegeneration}
\end{equation}
Here $n$ is a positive integer counting algebraically independent copies,
or flavours, of the Grassmann variables used in the standard
Berezin-integral representations of determinants and Pfaffians
\cite{Berezin1987,Efetov1997}.  No analytic continuation in $n$ is used
at this stage.  Put $m=2n$, attach a sign $\sigma_a$ to each replica index,
and define
\begin{equation}
 Z=\diag(z_+\1_n,z_-\1_n),\qquad
 T_{ab}=\begin{cases}1,&\sigma_a=\sigma_b,\\
 \alpha,&\sigma_a\ne\sigma_b.
 \end{cases}
 \label{eq:replicacovariance}
\end{equation}
The array $T=(T_{ab})$ distinguishes auxiliary-field entries joining
replicas at the same parameter value from those joining the two different
parameter sectors.  It enters the auxiliary Gaussian measure entrywise,
through the variances of the individual field components.  We take $0<\alpha\le1$, the independent case $\alpha=0$ following by a
degenerate-Gaussian limit.

For an auxiliary $m\times m$ matrix $Q$, define
\begin{equation}
 \cD_Z(Q)=\begin{pmatrix}Z&iQ\\iQ^\dagger&\bar Z\end{pmatrix}.
 \label{eq:sourceblock}
\end{equation}
For two complex $m\times m$ matrices $P,Q$ also define the antisymmetric matrix
\begin{equation}
 \cM_Z(P,Q)=
 \begin{pmatrix}
 0&iQ&-Z&iP\\
 -iQ^T&0&iP^\dagger&-\bar Z\\
 Z&-i\bar P&0&i\bar Q\\
 -iP^T&\bar Z&-i\bar Q^T&0
 \end{pmatrix},\qquad
 p_Z(P,Q)=(-1)^m\Pf\cM_Z(P,Q).
 \label{eq:AIpfblock}
\end{equation}
The fixed sign makes $p_Z(0,0)=|\det Z|^2$.

For $0<\alpha\leq1$, define the normalized complex and real
auxiliary-field measures
\[
 \begin{aligned}
 d\mu_T^{\C}(Q)
 &=
 \prod_{a,b=1}^{m}
 \frac{N}{\pi T_{ab}}
 \exp\left\{-\frac{N|Q_{ab}|^2}{T_{ab}}\right\}d^2Q_{ab},\\
 d\mu_T^{\R}(Q)
 &=
 \prod_{a,b=1}^{m}
 \sqrt{\frac{N}{\pi T_{ab}}}
 \exp\left\{-\frac{NQ_{ab}^2}{T_{ab}}\right\}dQ_{ab}.
 \end{aligned}
\]
Here $d^2Q_{ab}=d\Re Q_{ab}\,d\Im Q_{ab}$.  Under
$d\mu_T^{\C}$ the entries are independent and satisfy
$\E Q_{ab}=\E Q_{ab}^2=0$ and $\E|Q_{ab}|^2=T_{ab}/N$, whereas under
$d\mu_T^{\R}$ they are independent real Gaussians satisfying
$\E Q_{ab}^2=T_{ab}/(2N)$.

\begin{proposition}[Exact finite-$N$ auxiliary-field integral representations]
\label{prop:HS}
With the normalized measures defined above, the exact finite-dimensional
integral representations are
\begin{align}
 \cZ_{N,n}^{A}
 &=\int d\mu_T^{\C}(Q)\,\det\cD_Z(Q)^N,
 \label{eq:HS_A}\\
 \cZ_{N,n}^{\AiD}
 &=\int d\mu_T^{\C}(P)\,d\mu_T^{\C}(Q)\,p_Z(P,Q)^N,
 \label{eq:HS_AI}\\
 \cZ_{N,n}^{\AiiD}
 &=\int d\mu_T^{\R}(Q)\,\det\cD_Z(Q)^N.
 \label{eq:HS_AII}
\end{align}
  These are polynomial identities in the
spectral variables $z_\pm$ and $\bar z_\pm$. In the last line $Q^\dagger=Q^T$.
\end{proposition}

\begin{proof}
Set $z_a=z_{\sigma_a}$, where $\sigma_a=+$ for $1\leq a\leq n$
and $\sigma_a=-$ for $n<a\leq m=2n$.  Let $\eta$, $\bar\eta$, and
$\theta$ denote Grassmann-valued vectors: all their components
anticommute, and $\bar\eta$ is an independent set of variables rather
than the complex conjugate of $\eta$.  We fix the ordering of the
Berezin measures by the standard identities
\[
 \int \mathrm D(\bar\eta,\eta)\,
 e^{-\bar\eta^{T}H\eta}=\det H,
 \qquad
 \int \mathrm D\theta\,
 e^{\frac12\theta^{T}K\theta}=\Pf K,
\]
where $H$ is arbitrary and $K$ is antisymmetric, with dimensions
matching those of the corresponding Grassmann vectors
\cite{Berezin1987,Efetov1997}.

For classes $A$ and $\AiD$, introduce Grassmann variables
$\bar\psi_i^a,\psi_i^a,\bar\chi_i^a,\chi_i^a$.  Before averaging over
the matrices, the product of determinants is
\begin{align}
 &\prod_{a=1}^{m}
 \det\!\bigl(z_a-X_{\sigma_a}\bigr)
 \det\!\bigl(\bar z_a-X_{\sigma_a}^{\dagger}\bigr)
 \notag\\
 &\quad =
 \int \mathrm D(\bar\psi,\psi,\bar\chi,\chi)\,
 \exp\Biggl\{
 -\sum_{i,a}\bigl(
   \bar\psi_i^a z_a\psi_i^a+
   \bar\chi_i^a\bar z_a\chi_i^a
 \bigr)
 \notag\\
 &\hspace{42mm}
 +\sum_{i,j,a}\left[
   \bar\psi_i^a(X_{\sigma_a})_{ij}\psi_j^a+
   \bar\chi_j^a\overline{(X_{\sigma_a})_{ij}}\chi_i^a
 \right]\Biggr\}.
 \label{eq:GrassmannA_AI}
\end{align}
Define the even Grassmann bilinears
\[
 A_{ab}=\sum_i\bar\psi_i^a\chi_i^b,\qquad
 B_{ba}=\sum_i\bar\chi_i^b\psi_i^a,\qquad
 C_{ab}=\sum_i\bar\psi_i^a\bar\chi_i^b,\qquad
 D_{ab}=\sum_i\psi_i^a\chi_i^b.
\]
Using the covariance of the correlated matrices and then reordering the
Grassmann variables gives
\[
 \E\exp\{\text{matrix-dependent part of \eqref{eq:GrassmannA_AI}}\}
 =
 \exp\left\{
 -\frac1N\sum_{a,b}T_{ab}A_{ab}B_{ba}
 \right\}
\]
in class $A$, and
\[
 \E\exp\{\text{matrix-dependent part of \eqref{eq:GrassmannA_AI}}\}
 =
 \exp\left\{
 -\frac1N\sum_{a,b}T_{ab}
 \bigl(A_{ab}B_{ba}+C_{ab}D_{ab}\bigr)
 \right\}
\]
in class $\AiD$.  The second quartic invariant is precisely the
contribution of the exchange term in \eqref{eq:AIcov}.

For commuting nilpotent variables $A,B$, the required complex
Hubbard--Stratonovich transformation is
\begin{equation}
 \frac{N}{\pi t}\int_{\C}
 \exp\left\{-\frac{N|q|^2}{t}+iqA+i\bar qB\right\}\,d^2q
 =
 \exp\left\{-\frac{t}{N}AB\right\},
 \qquad t>0.
 \label{eq:complexHSscalar}
\end{equation}
Applying \eqref{eq:complexHSscalar} with $t=T_{ab}$ to every pair
$(a,b)$ produces precisely the normalized measure $d\mu_T^{\C}$
defined above.

In class $A$, application of \eqref{eq:complexHSscalar} leaves, at each
physical index $i$, the Grassmann integral
\begin{align}
 &\int \mathrm D(\bar\psi_i,\psi_i,\bar\chi_i,\chi_i)
 \exp\Bigl\{
 -\bar\psi_i^TZ\psi_i-\bar\chi_i^T\bar Z\chi_i
 +i\bar\psi_i^TQ\chi_i+i\bar\chi_i^TQ^\dagger\psi_i
 \Bigr\}
 \notag\\
 &\qquad =\det\cD_Z(-Q).
 \label{eq:localAintegral}
\end{align}
The Gaussian law of $Q$ is invariant under $Q\mapsto-Q$, so the sign in
\eqref{eq:localAintegral} is immaterial.  Since the integral factorizes
over $i=1,\ldots,N$, it gives
$\cZ_{N,n}^{A}(z_+,z_-;\alpha)=\E_Q\det\cD_Z(Q)^N$, which is
\eqref{eq:HS_A}.

For class $\AiD$, use a field $P$ to decouple $A_{ab}B_{ba}$ and an
independent field $Q$ to decouple $C_{ab}D_{ab}$.  Both have measure
$d\mu_T^{\C}$.  With
\[
 \Theta_i=
 \begin{pmatrix}
  \bar\psi_i\\ \bar\chi_i\\ \psi_i\\ \chi_i
 \end{pmatrix},
\]
the complete one-site exponent after the two Hubbard--Stratonovich
transformations is
\begin{align}
 &-\bar\psi_i^TZ\psi_i-\bar\chi_i^T\bar Z\chi_i
 +i\bar\psi_i^TQ\bar\chi_i+i\psi_i^T\bar Q\chi_i
 \notag\\
 &\hspace{31mm}
 +i\bar\psi_i^TP\chi_i+i\bar\chi_i^TP^\dagger\psi_i
 =
 \frac12\Theta_i^T\cM_Z(P,Q)\Theta_i .
 \label{eq:localAIexponent}
\end{align}
Consequently,
\[
 \int \mathrm D\Theta_i\,
 e^{\frac12\Theta_i^T\cM_Z(P,Q)\Theta_i}
 =
 (-1)^m\Pf\cM_Z(P,Q)
 =
 p_Z(P,Q).
\]
Here the factor $(-1)^m$ comes from permuting the Berezin variables from
the ordering inherited from the two determinant representations to
$(\bar\psi,\bar\chi,\psi,\chi)$.  It is also the factor that ensures
$p_Z(0,0)=|\det Z|^2$. Factorization over the $N$ physical indices now yields
$\cZ_{N,n}^{\AiD}(z_+,z_-;\alpha)
=\E_{P,Q}\!\left[p_Z(P,Q)^N\right]$, proving
\eqref{eq:HS_AI}.

It remains to treat class $\AiiD$.  Set $Y=X\mathbb J$, so that
$Y^T=-Y$.  For $i<j$, the entries $Y_{ij}$ are independent complex
Gaussians of variance $1/(2N)$, and their correlations between the two
parameters are encoded by $T_{ab}$.  The monic Pfaffian polynomial has
the Berezin representation
\[
 P_{\AiiD}(z;X)
 =
 \frac{1}{\Pf\mathbb J}
 \int\mathrm D\psi\,
 \exp\left\{
 \frac12\psi^T(z\mathbb J-Y)\psi
 \right\}.
\]
A second set $\chi$ represents its complex conjugate.  Hence the part
of the replicated exponent containing $Y$ is
\[
 -\sum_{a=1}^{m}\sum_{i<j}
 \left[
  (Y_{\sigma_a})_{ij}\psi_i^a\psi_j^a+
  \overline{(Y_{\sigma_a})_{ij}}\chi_i^a\chi_j^a
 \right].
\]
Gaussian averaging gives
\begin{align}
 &\E\exp\left\{
 -\sum_{a,i<j}
 \left[
  (Y_{\sigma_a})_{ij}\psi_i^a\psi_j^a+
  \overline{(Y_{\sigma_a})_{ij}}\chi_i^a\chi_j^a
 \right]\right\}
  =
 \exp\left\{
 \frac1{2N}\sum_{a,b}T_{ab}
 \sum_{i<j}\psi_i^a\psi_j^a\chi_i^b\chi_j^b
 \right\}.
 \label{eq:AIIaverageexplicit}
\end{align}
Introducing $F_{ab}=\sum_{i=1}^{2N}\psi_i^a\chi_i^b$, one has
$F_{ab}^{\,2}=-2\sum_{i<j}\psi_i^a\psi_j^a\chi_i^b\chi_j^b$.
Thus \eqref{eq:AIIaverageexplicit} is
\[
 \exp\left\{
 -\frac1{4N}\sum_{a,b}T_{ab}F_{ab}^{\,2}
 \right\}.
\]

The required real Hubbard--Stratonovich transformation is
\begin{equation}
 \sqrt{\frac{N}{\pi t}}\int_{\R}
 \exp\left\{-\frac{Nq^2}{t}+iqF\right\}\,dq
 =
 \exp\left\{-\frac{t}{4N}F^2\right\},
 \qquad t>0.
 \label{eq:realHSscalar}
\end{equation}
Applying \eqref{eq:realHSscalar} to every pair $(a,b)$ produces the
normalized measure $d\mu_T^{\R}$ defined above and replaces the quartic interaction by
$i\sum_{a,b}Q_{ab}F_{ab}$.

Finally, write $p=2j-1$ and $q=2j$ for the two coordinates in the
$j$th physical block of $\mathbb J$.  Its local Grassmann integral is
\begin{align}
 I_j(Q)
 &=
 \int \mathrm D(\psi_p,\psi_q,\chi_p,\chi_q)
 \exp\Biggl\{
 \sum_a\left(
  z_a\psi_p^a\psi_q^a+
  \bar z_a\chi_p^a\chi_q^a
 \right)
 \notag\\
 &\hspace{37mm}
 +i\sum_{a,b}Q_{ab}
 \left(\psi_p^a\chi_p^b+\psi_q^a\chi_q^b\right)
 \Biggr\}.
 \label{eq:AIIlocalintegral}
\end{align}
Indeed, upon defining the two independent $2m$-component Grassmann
vectors
\[
 \eta_j=
 \begin{pmatrix}\psi_q\\ \chi_p\end{pmatrix},
 \qquad
 \widehat\eta_j^{\,T}=
 \begin{pmatrix}\psi_p^T&-\chi_q^T\end{pmatrix},
\]
the exponent in \eqref{eq:AIIlocalintegral} becomes
$\widehat\eta_j^{\,T}\cD_Z(Q)\eta_j$.  The determinant Berezin identity
therefore gives $I_j(Q)=\det\cD_Z(Q)$.
The normalization by $\Pf\mathbb J$ fixes the accompanying Pfaffian
signs and makes the polynomial monic.  Since there are $N$ independent
physical $2\times2$ blocks,
$\prod_{j=1}^{N}I_j(Q)=\det\cD_Z(Q)^N$.
Averaging this expression with respect to $d\mu_T^{\R}(Q)$ proves
\eqref{eq:HS_AII}, including the real auxiliary-field variance
$T_{ab}/(2N)$ and the power $N$.
\end{proof}

At $\alpha=1$ these identities specialize to static
characteristic-polynomial averages in the family studied in
\cite{AkemannAygunKieburgPassler2025,Forrester2025}.  
The new finite-$N$ ingredient is the entrywise covariance table $T_{ab}$
joining the two parameter sectors.  Under saddle-point reduction, this
covariance produces the $a$-dependent term in the compact-group action
displayed in \eqref{eq:masscalculation}.

\subsection{A common radial action and the correct replica dimensions}

For class $\AiD$ combine its two fields into
\begin{equation}
 \cQ=\begin{pmatrix}P&Q\\-\bar Q&\bar P\end{pmatrix},\qquad
 \Jrep=\begin{pmatrix}0&\1_m\\-\1_m&0\end{pmatrix},\qquad
 \cQ=\Jrep\bar{\cQ}\Jrep^{-1}.
 \label{eq:quaternionfield}
\end{equation}
Thus $\cQ$ is an $m\times m$ quaternion matrix in its $2m\times2m$ complex
representation. 
\begin{equation}
 p_Z(P,Q)^2=\det\cD_{\widetilde Z}(\cQ),
 \qquad \widetilde Z=\diag(Z,Z).
 \label{eq:AIsquareddeterminant}
\end{equation}
This squared identity is useful for the subsequent saddle-point
evaluation, whereas the unsquared finite-$N$ object remains the polynomial
defined in \eqref{eq:AIpfblock}.

Let $r_X$ denote the row dimension, in the ordinary complex
representation, of the matrix appearing in the common saddle-point action
\eqref{eq:radialaction} below: this matrix is $Q$ in classes $A$ and $\AiiD$,
and $\cQ$ in class $\AiD$.  Set
$\gamma_A=\gamma_{\AiiD}=1$ and $\gamma_{\AiD}=1/2$.  Throughout,
$\Tr$ denotes the ordinary complex matrix trace, including when $\cQ$
is the complex representation of a quaternion matrix.  Thus the
prefactor $\gamma_XN$ in \eqref{eq:radialaction} equals $N$ in classes
$A$ and $\AiiD$, and $N/2$ in class $\AiD$.
\begin{center}
\begin{tabular}{c c c c c}
\toprule
class & field & $r_X$ & $\gamma_X$ & compact group $G_X$\\
\midrule
$A$ & complex & $2n$ & $1$ & $U(2n)$\\
$\AiD$ & quaternion & $4n$ & $1/2$ & $Sp(2n)\subset U(4n)$\\
$\AiiD$ & real & $2n$ & $1$ & $O(2n)$\\
\bottomrule
\end{tabular}
\end{center}
In particular, $\gamma_Xr_X=2n$ in all three cases. The two spectral variables occupy the diagonal blocks of
$Z=\diag(z_+\1_n,z_-\1_n)$ and thereby determine the decomposition
$\C^{2n}=\C_+^n\oplus\C_-^n$, represented by the block-sign matrix
$\Lambda=\diag(\1_n,-\1_n)$.  Let $\Lambda_X$ denote the corresponding
matrix in the auxiliary-field space.  Then $\Lambda_X=\Lambda$ in classes
$A$ and $\AiiD$, whereas the doubled complex representation used for
class $\AiD$ requires
$\Lambda_{\AiD}=\diag(\Lambda,\Lambda)$.  This doubling is necessary for every trace in that case to have the
correct dimension.

At $\alpha=1$ and for coincident sources $z_+=z_-=z_0$, the exact
auxiliary-field representations take the common form
\begin{equation}
 \begin{aligned}
 \cZ_{N,n}^X(z_0,z_0;1)
 &=
 C_{N,n}^X
 \int
 \exp\left\{-\gamma_XN\Tr\left[
 QQ^\dagger-\log\bigl(|z_0|^2\1_{r_X}+QQ^\dagger\bigr)
 \right]\right\}\dd Q\\
 &=C_{N,n}^X\int e^{-S_0(Q)}\dd Q ,
 \end{aligned}
 \label{eq:radialaction}
\end{equation}
where $C_{N,n}^X$ is independent of $z_0$, and the integral is over the
independent components of the auxiliary field listed in the table.  In
class $\AiD$, $Q$ in this common notation denotes the doubled complex
matrix $\cQ$ defined in \eqref{eq:quaternionfield}.
Writing $h^2$ for an eigenvalue of $QQ^\dagger$, stationarity in the
radial direction gives $1-(|z_0|^2+h^2)^{-1}=0$.  Thus every singular
value of $Q$ equals $\sqrt{\chi_0}$, and the saddle-point manifold is
parametrized by
\begin{equation}
 QQ^\dagger=\chi_0\1_{r_X},\qquad
 Q=\sqrt{\chi_0}\,U,\qquad U\in G_X,\qquad
 \chi_0=1-|z_0|^2>0.
 \label{eq:radialsaddle}
\end{equation}
The resulting angular orbit gives the familiar compact Grassmannian after
quotienting by the stabilizer of $\Lambda_X$.  We retain the full group with
normalized Haar measure, which avoids separate quotient-volume constants.

\subsection{The covariance deformation on the saddle}

In this subsection $Q$ denotes the auxiliary matrix in classes $A$ and
$\AiiD$, and the complex representation $\cQ$ in class $\AiD$.  Write
$\lambda_\mu=(\Lambda_X)_{\mu\mu}\in\{+1,-1\}$ for
$1\leq\mu\leq r_X$.
For class $\AiD$, each original replica label occurs twice in
$\Lambda_X=\diag(\Lambda,\Lambda)$.  The covariance table
\eqref{eq:replicacovariance}, lifted to these complex indices, is therefore
\[
 T_{\mu\nu}^{(X)}
 =
 \begin{cases}
  1, & \lambda_\mu\lambda_\nu=1,\\
  \alpha, & \lambda_\mu\lambda_\nu=-1.
 \end{cases}
\]
By the Gaussian auxiliary-field measures in \cref{prop:HS}, with $P,Q$
combined into $\cQ$ according to \eqref{eq:quaternionfield} in class
$\AiD$, the Gaussian part of the action is
\begin{equation}
 \begin{aligned}
 S_{\mathrm G}^{(\alpha)}(Q)
 &=
 \gamma_XN\sum_{\mu,\nu=1}^{r_X}
 \frac{|Q_{\mu\nu}|^2}{T_{\mu\nu}^{(X)}}\\
 &=
 \gamma_XN\Tr(QQ^\dagger)
 +\gamma_XN(\alpha^{-1}-1)
 \sum_{\lambda_\mu\lambda_\nu=-1}|Q_{\mu\nu}|^2 .
 \end{aligned}
 \label{eq:timeaction}
\end{equation}
At $\alpha=1$, the second term vanishes and the first is precisely the
Gaussian term in the radial action \eqref{eq:radialaction}.  The
determinant or Pfaffian factor is independent of $\alpha$, so that at
coincident sources the full action is
$S^{(\alpha)}(Q)=S_0(Q)+\delta S_\alpha(Q)$, where
$\delta S_\alpha(Q)$ is the second term in \eqref{eq:timeaction}.

Here $|Q_{\mu\nu}|^2=Q_{\mu\nu}^2$ for the real field in class
$\AiiD$.  For the quaternion field, $\gamma_{\AiD}=1/2$ compensates for
the doubling in its ordinary complex representation.  The
$Q$-independent change in the normalized Gaussian prefactor is
$O(n^2/N)=o(1)$ at fixed $n$ and is suppressed.
The sum over entries joining the two replica sectors can be expressed as a
trace: $\Tr(QQ^\dagger)=\sum_{\mu,\nu}|Q_{\mu\nu}|^2$, whereas
$\Tr(\Lambda_XQ\Lambda_XQ^\dagger)=
\sum_{\mu,\nu}\lambda_\mu\lambda_\nu|Q_{\mu\nu}|^2$.
It follows that
\begin{align}
 \sum_{\lambda_\mu\lambda_\nu=-1}|Q_{\mu\nu}|^2
 &=
 \frac12\sum_{\mu,\nu}
 (1-\lambda_\mu\lambda_\nu)|Q_{\mu\nu}|^2
=\frac12\left[\Tr(QQ^\dagger)-\Tr(\Lambda_XQ\Lambda_XQ^\dagger)
 \right].
 \label{eq:crossblocktrace}
\end{align}

On the saddle-point manifold \eqref{eq:radialsaddle}, $Q=\sqrt{\chi_0}\,U$, so
$\Tr(QQ^\dagger)=\chi_0r_X$ and, by cyclicity,
$\Tr(\Lambda_XQ\Lambda_XQ^\dagger)=
\chi_0\Tr(U\Lambda_XU^\dagger\Lambda_X)$.
Substitution into \eqref{eq:timeaction} therefore gives
\[
 \delta S_\alpha
 =
 \frac{\gamma_XN\chi_0}{2}(\alpha^{-1}-1)
 \left[
  r_X-\Tr(U\Lambda_XU^\dagger\Lambda_X)
 \right].
\]
Since $\alpha=1-v^2/(2N)+O(N^{-2})$,
$\alpha^{-1}-1=v^2/(2N)+O(N^{-2})$, and $a=v^2\chi_0$,
we obtain
\begin{equation}
 \delta S_\alpha
 =
 \frac{\gamma_Xa}{4}
 \left[
  r_X-\Tr(U\Lambda_XU^\dagger\Lambda_X)
 \right]+o(1).
 \label{eq:masscalculation}
\end{equation}
Using $\gamma_Xr_X=2n$, this may equivalently be written as
\[
 \delta S_\alpha
 =
 \frac{na}{2}
 -
 \frac{\gamma_Xa}{4}
 \Tr(U\Lambda_XU^\dagger\Lambda_X)
 +o(1).
\]
Thus the covariance deformation contributes to the auxiliary-field
integrand the factor
\[
 e^{-\delta S_\alpha}
 =
 e^{-na/2}
 \exp\left\{
 \frac{\gamma_Xa}{4}
 \Tr(U\Lambda_XU^\dagger\Lambda_X)
 \right\}
 [1+o(1)].
\]
The radial-fluctuation calculation in the next subsection gives the spatial
factor
\[
 \exp\left\{
 \frac{\gamma_Xu}{4}
 \Tr(U\Lambda_XU^\dagger\Lambda_X)
 \right\};
\]
see \eqref{eq:spatialmass}.  The covariance and spatial deformations
therefore combine into
\[
 e^{-na/2}
 \exp\left\{
 \frac{\gamma_X(u+a)}{4}
 \Tr(U\Lambda_XU^\dagger\Lambda_X)
 \right\},
\]
which establishes the appearance of the single microscopic parameter
$s=u+a$.

\subsection{Why the spatial coefficient is independent of \texorpdfstring{$\chi_0$}{chi0}}

For completeness we also perform the radial-fluctuation calculation that
fixes the spatial coefficient. Circular symmetry allows us to choose $z_0=b\geq0$, so that
$b^2+\chi_0=1$.  In a neighborhood of the radial saddle point manifold
\eqref{eq:radialsaddle}, we use the convenient parametrization
\begin{equation}
 Q=\bigl(\sqrt{\chi_0}\,\1+N^{-1/2}H\bigr)U,
 \qquad
 V=U\Lambda_XU^\dagger,
 \qquad
 \epsilon=\frac{\omega}{2\sqrt N},
 \label{param}
 \end{equation}
where $U\in G_X$ parametrizes the angular saddle directions and $H$ is
Hermitian over the real, complex, or quaternion field, as appropriate,
and parametrizes the radial fluctuations. 

To identify the matrix being expanded, define
$Z_X=b\1_{r_X}+\epsilon\Lambda_X$ with
$\epsilon=\omega/(2\sqrt N)$.
Here $Z_X=Z$ in classes $A$ and $\AiiD$, whereas
$Z_X=\widetilde Z$ in class $\AiD$.  The matrix whose determinant
produces the logarithmic term in \eqref{eq:radialaction} is therefore the
block $\cD_{Z_X}(Q)$ defined in \eqref{eq:sourceblock}, and in class
$\AiD$ this follows from \eqref{eq:AIsquareddeterminant}.

Substituting the local saddle parametrization \eqref{param}
into \eqref{eq:sourceblock} gives
\[
 \cD_{Z_X}(Q)
 =
 \begin{pmatrix}
  b\1+\epsilon\Lambda_X&
  i\bigl(\sqrt{\chi_0}\,\1+N^{-1/2}H\bigr)U\\
  iU^\dagger\bigl(\sqrt{\chi_0}\,\1+N^{-1/2}H\bigr)&
  b\1+\bar\epsilon\Lambda_X
 \end{pmatrix}.
\]
Introduce the block-diagonal matrix
\[
 \mathcal U=
 \begin{pmatrix}
  \1&0\\
  0&U^\dagger
 \end{pmatrix},
 \qquad
 V=U\Lambda_XU^\dagger.
\]
Conjugation by $\mathcal U$ leaves the determinant unchanged and removes
$U$ from the off-diagonal blocks:
\begin{align}
 \mathcal U^\dagger\cD_{Z_X}(Q)\mathcal U
 &=
 \begin{pmatrix}
  b\1+\epsilon\Lambda_X&
  i\bigl(\sqrt{\chi_0}\,\1+N^{-1/2}H\bigr)\\
  i\bigl(\sqrt{\chi_0}\,\1+N^{-1/2}H\bigr)&
  b\1+\bar\epsilon V
 \end{pmatrix}
 =M_0+\delta M.
 \label{eq:sourceblockexpansion}
\end{align}
On the unperturbed saddle-point manifold, $\epsilon=0$ and $H=0$, so
\[
 M_0=
 \begin{pmatrix}
  b&i\sqrt{\chi_0}\\
  i\sqrt{\chi_0}&b
 \end{pmatrix}\otimes\1_{r_X}.
\]
Since $b^2+\chi_0=1$, its inverse is
\[
 M_0^{-1}=
 \begin{pmatrix}
  b&-i\sqrt{\chi_0}\\
  -i\sqrt{\chi_0}&b
 \end{pmatrix}\otimes\1_{r_X}.
\]
Expand $\Tr\log(M_0+\delta M)$ through second order, using
$\Tr\Lambda_X=\Tr V=0$.  The linear radial terms cancel between the Gaussian
weight and the logarithm.  The part of the action that can depend on $U$ is
\begin{align}
 S-S_0={}&2\gamma_X\chi_0\Tr H^2
 +2\gamma_Xb\sqrt{\chi_0N}\Tr[H(\epsilon\Lambda_X+\bar\epsilon V)]
 \nonumber\\
 &+\frac{\gamma_XNb^2r_X}{2}(\epsilon^2+\bar\epsilon^2)
 -\gamma_XN\chi_0|\epsilon|^2\Tr(\Lambda_XV)+o(1).
 \label{eq:radialcompletion}
\end{align}
Completing the square in the radial fluctuation $H$ and then performing
the Gaussian integral over $H$ produces, in the exponent of the remaining
integral over $U$, the additional term
\[
 \frac{\gamma_XNb^2}{2}
 \Tr\left[(\epsilon\Lambda_X+\bar\epsilon V)^2\right].
\]
Since $\Tr\Lambda_X^2=\Tr V^2=r_X$, the parts proportional to
$\epsilon^2+\bar\epsilon^2$ precisely cancel the corresponding terms from
\eqref{eq:radialcompletion} when $e^{-(S-S_0)}$ is formed.  The mixed part is $\gamma_XNb^2|\epsilon|^2\Tr(\Lambda_XV)$, whereas the
last term in \eqref{eq:radialcompletion}, upon passing from the action to
the exponent, contributes $\gamma_XN\chi_0|\epsilon|^2\Tr(\Lambda_XV)$,
their sum having the same form with $\chi_0$ replaced by $\chi_0+b^2=1$. Consequently the spatial contribution to
the angular exponent is
\begin{equation}
 \frac{\gamma_Xu}{4}\Tr(U\Lambda_XU^\dagger\Lambda_X),
 \qquad u=N|z_+-z_-|^2.
 \label{eq:spatialmass}
\end{equation}
\begin{remark}[Complexified spectral variables]
Here $z_\pm=b\pm\epsilon$, where
$\epsilon=\omega/(2\sqrt N)$ is the complex microscopic displacement from
the bulk point $z_0=b$.  Density correlations are obtained by
differentiating with respect to $z_\pm$ and $\bar z_\pm$.  For this purpose
we temporarily regard $\epsilon$ and $\bar\epsilon$ as independent complex
variables. The exact
finite-dimensional auxiliary-field integrals, as well as the Gaussian
integral over $H$ appearing in their saddle-point expansion, are entire
in these variables.  Moreover, the radial Hessian is
$4\gamma_X\chi_0$ in the trace metric and is bounded away from zero on
compact subsets of the bulk, where $\chi_0=1-|z_0|^2>0$.  Together with
the uniform saddle-point estimates below, these facts justify taking
source derivatives after the expansion.
\end{remark}

\begin{remark}
For class $A$, the same cancellation was obtained in singular-value
coordinates in our earlier work with Lacroix-A-Chez-Toine, see
\cite[Supplemental Material, Eqs.~(S13)--(S21)]{FyodorovLacroix2026}.
The calculation above recasts that argument as a radial Gaussian completion
and exhibits it in a form that applies uniformly to classes $A$, $\AiD$,
and $\AiiD$.
\end{remark}

\begin{proposition}[Fixed-integer-replica bulk reduction]
\label{prop:fixedreplica}
Fix $n\ge1$.  Uniformly for $z_0$ in a compact subset of the open unit disk
and $\omega,v$ in compact sets, let $\alpha_N=1-v^2/(2N)+O(N^{-2})$ with a
uniform remainder.  Then
\begin{equation}
 \frac{\cZ_{N,n}^X(z_0+\omega/(2\sqrt N),z_0-\omega/(2\sqrt N);\alpha_N)}
 {\cZ_{N,n}^X(z_0,z_0;1)}
 =e^{-na/2}\cI_n^X(u+a)[1+o(1)],
 \label{eq:fixedreplicareduction}
\end{equation}
where
\begin{equation}
 \cI_n^X(s)=\int_{G_X}
 \exp\left\{\frac{\gamma_Xs}{4}
 \Tr(U\Lambda_XU^\dagger\Lambda_X)\right\}dU,\qquad \int_{G_X}dU=1.
 \label{eq:allcompact}
\end{equation}
The convergence is locally uniform when $z_\pm$ and $\bar z_\pm$ are
regarded as independent complex variables.  It therefore remains valid
after any fixed number of derivatives with respect to these variables. in particular after the four
derivatives used for the density correlation.
\end{proposition}

\begin{proof}
We now justify the large-$N$ reduction at fixed positive integer $n$ by
a finite-dimensional saddle-point (Laplace) analysis, whereas the
analytic continuation in $n$ enters only later.  By rotational invariance in the
spectral plane, take $z_0=b\geq0$ and set $\chi_0=1-b^2$.  At coincident spectral
sources and $\alpha=1$, the determinant factors in classes $A$ and
$\AiiD$ reduce to powers of
$\det(b^2\1_{r_X}+QQ^\dagger)$ and are therefore nonnegative.  In class
$\AiD$, \eqref{eq:AIsquareddeterminant} gives
$p_Z(P,Q)^2=\det(b^2\1_{r_{\AiD}}+\cQ\cQ^\dagger)$.  For $b>0$ this
determinant is strictly positive: since the auxiliary-field space is
connected and $p_Z(0,0)=|\det Z|^2>0$, the polynomial $p_Z(P,Q)$ is its
positive square root.  The case $b=0$ follows by continuity.  Thus all
three classes have the common radial action \eqref{eq:radialaction}.

On the full-rank set use the polar decomposition $Q=HU$, where $H$ is
positive self-adjoint over the relevant real, complex, or quaternion field
and $U\in G_X$.  This set has full measure.  The map $(H,U)\mapsto HU$ is
smooth near $H=\sqrt{\chi_0}\,\1$, and Lebesgue measure takes the form
$J_X(H)\,dH\,dU$, with $J_X$ smooth and strictly positive there.  Keeping
$H$ as a matrix, rather than diagonalizing it, avoids a spurious
Vandermonde singularity at the scalar saddle.  For $z_0$ in a fixed compact
subset of the disk, $\chi_0$ is bounded below.  The scalar function
$h^2-\log(b^2+h^2)$ has its unique minimum at
$h=\sqrt{\chi_0}$, a uniformly positive Hessian, and a positive gap outside
every fixed neighborhood of that point.  Its quadratic growth controls the
tails.  Hence the complement of a fixed tube around the compact saddle
manifold is exponentially small relative to the coincident-source integral,
uniformly in the stated bulk regime.

Inside the tube write
$H=\sqrt{\chi_0}\,\1+N^{-1/2}K$.  For
$\|K\|\le N^{1/10}$, Taylor expansion of the action and the polar Jacobian
has a uniform remainder bounded by
$C N^{-1/2}(1+\|K\|^3)$.  Outside this growing set, but inside the tube, the
positive radial Hessian supplies an integrable Gaussian majorant.  The
source calculation \eqref{eq:radialcompletion} is therefore uniform under
the radial integral.  Completing its Gaussian square gives precisely the
angular factor \eqref{eq:spatialmass}.  Its remaining determinant and the
value $J_X(\sqrt{\chi_0}\,\1)$ do not depend on $U$ and cancel against the
denominator in \eqref{eq:fixedreplicareduction}.  Since
$\alpha_N^{-1}-1=O(N^{-1})$, replacing $H$ by its saddle value in
\eqref{eq:timeaction} makes an $o(1)$ error under the same Gaussian
majorant.  Equation \eqref{eq:masscalculation} then produces the factor
$e^{-na/2}$ and adds
$(\gamma_Xa/4)\Tr(U\Lambda_XU^\dagger\Lambda_X)$ to the exponent of the
compact-group integral. All
components of $O(2n)$ are included by the real polar decomposition.

The estimates remain valid when the spectral sources $z_\pm$ and
$\bar z_\pm$ are regarded as independent complex variables and varied in
a fixed neighborhood of their physical values.  The finite-dimensional
integrands are entire in these variables and, after slightly shrinking the
neighborhood, the preceding majorant applies uniformly.  Local
uniform convergence and Cauchy's formula then give convergence of every
fixed source derivative.  Finally, normalized Haar integration is over a
compact group, so the estimates are uniform in $U$.  This proves the stated
asymptotic and its differentiated form.
\end{proof}

The coefficients in the two new classes are therefore
\begin{align}
 \cI_n^{\AiD}(u+a)
 &=\int_{Sp(2n)}e^{(u+a)\Tr(U\Lambda_{\AiD}U^\dagger\Lambda_{\AiD})/8}
 \,dU,\label{eq:AIsigma}\\
 \cI_n^{\AiiD}(u+a)
 &=\int_{O(2n)}e^{(u+a)\Tr(U\Lambda U^T\Lambda)/4}
 \,dU.\label{eq:AIIsigma}
\end{align}
At $a=0$ these agree with the static integrals of
\cite{ChenXiaoLiuRyu2026}.  The calculation above identifies the coefficient
of $a$ independently of the short-time overlap argument.

\subsection{From the group integral to a density: the replica step}
\label{sec:replica}

At finite $N$, put $L_\sigma=\log|P_X(z_\sigma;X_\sigma)|^2, \, \sigma=\pm$.  For real
$n$ near zero, the original expectation in \eqref{eq:finitegeneration} has
the cumulant expansion (the Taylor expansion in connected moments)
\[
 \log\cZ_{N,n}
 =n\E(L_++L_-)+\frac{n^2}{2}\Var(L_++L_-)+O(n^3).
\]
Expanding the variance gives
$\Var(L_++L_-)=\Var(L_+)+\Var(L_-)+2\Cov(L_+,L_-)$.
Now $L_+$ depends only on $(z_+,\bar z_+)$ and $L_-$ only on
$(z_-,\bar z_-)$.  Therefore
\[
 \partial_{z_+}\partial_{\bar z_+}
 \partial_{z_-}\partial_{\bar z_-}\Var(L_+)=0,
 \qquad
 \partial_{z_+}\partial_{\bar z_+}
 \partial_{z_-}\partial_{\bar z_-}\Var(L_-)=0,
\]
whereas
\[
 \frac12
 \partial_{z_+}\partial_{\bar z_+}
 \partial_{z_-}\partial_{\bar z_-}\Var(L_++L_-)
 =
 \partial_{z_+}\partial_{\bar z_+}
 \partial_{z_-}\partial_{\bar z_-}\Cov(L_+,L_-).
\]
Consequently, after applying these derivatives the coefficient
of $n^2$ in $\log\cZ_{N,n}$ is precisely the differentiated covariance
$\Cov(L_+,L_-)$. The one-dimensional Poincar\'e--Lelong identity
$\partial_z\partial_{\bar z}\log|z-\lambda|^2
=\pi\delta^{(2)}(z-\lambda)$, understood in the distributional sense,
converts the covariance of logarithmic spectral polynomials into the
density covariance. It is the same identity whether a determinant or a
monic Pfaffian counts the points.

The exact finite-$N$ cumulant expansion does not justify taking $n\to0$ in
\cref{prop:fixedreplica}: that proposition only holds at each fixed positive
integer.  We now adopt, for the $n\to0$ replica limit, the analytic-continuation
prescription proposed in the static duality calculation
\cite{ChenXiaoLiuRyu2026}, namely the continuation obtained by matching
to the corresponding Hermitian correlation function in the upper half-plane.
  Apply the static replica continuation to $\cI_n^X(s)$ in
\eqref{eq:allcompact} and expand its logarithm near $n=0$:
\begin{equation}
 \log\cI_n^X(s)
 =
 n\,\phi_{X,1}(s)+n^2\Phi_X(s)+O(n^3).
 \label{eq:replicaexpansion}
\end{equation}
Equivalently, the separation-dependent connected coefficient is
$\Phi_X(s)=\left.\tfrac12\partial_n^2\log\cI_n^X(s)\right|_{n=0}$,
where the derivative with respect to $n$ is understood in the analytic
continuation selected by the static duality calculation
\cite{ChenXiaoLiuRyu2026}.  The relation between $\Phi_X$ and the memory kernel introduced in
\eqref{eq:threeK} follows from the density insertions. Each density insertion at $z_\sigma$, $\sigma\in\{+,-\}$, is generated by
$\partial_{z_\sigma}\partial_{\bar z_\sigma}$, where $z_\sigma$ and
$\bar z_\sigma$ are treated as independent Wirtinger variables.  The two
density insertions therefore produce the mixed operator
$\partial_{z_+}\partial_{\bar z_+}\partial_{z_-}\partial_{\bar z_-}$.
To pass to microscopic coordinates, write
$z_0=(z_++z_-)/2$ and $\omega=\sqrt N\,(z_+-z_-)$,
so that
\[
 \partial_{z_+}
 =
 \frac12\partial_{z_0}+\sqrt N\,\partial_\omega,
 \qquad
 \partial_{z_-}
 =
 \frac12\partial_{z_0}-\sqrt N\,\partial_\omega,
\]
with the analogous identities
\[
 \partial_{\bar z_+}
 =
 \frac12\partial_{\bar z_0}+\sqrt N\,\partial_{\bar\omega},
 \qquad
 \partial_{\bar z_-}
 =
 \frac12\partial_{\bar z_0}-\sqrt N\,\partial_{\bar\omega}.
\]
Consequently,
\[
 \partial_{z_+}\partial_{z_-}
 =
 -N\partial_\omega^2+\frac14\partial_{z_0}^2,
 \qquad
 \partial_{\bar z_+}\partial_{\bar z_-}
 =
 -N\partial_{\bar\omega}^2+\frac14\partial_{\bar z_0}^2.
\]
The microscopic density is $N^{-1}$ times the density in the original
spectral coordinates, so a two-density covariance carries the prefactor
$N^{-2}$.  Hence
\begin{align}
 &N^{-2}
 \partial_{z_+}\partial_{\bar z_+}
 \partial_{z_-}\partial_{\bar z_-}
 \notag\\
 &\qquad =
 \partial_\omega^2\partial_{\bar\omega}^2
 -\frac{1}{4N}\left(
  \partial_\omega^2\partial_{\bar z_0}^2+
  \partial_{z_0}^2\partial_{\bar\omega}^2
 \right)
 +\frac{1}{16N^2}
  \partial_{z_0}^2\partial_{\bar z_0}^2.
 \label{eq:microscopicderivatives}
\end{align}
For a bulk scaling function whose dependence on $z_0$ remains smooth, the
last two terms vanish as $N\to\infty$.  Thus the two microscopic density
insertions act on the separation-dependent part as
$\partial_\omega^2\partial_{\bar\omega}^2$.  Since
$s=|\omega|^2+a=u+a$, any twice differentiable radial function satisfies
\begin{equation}
 \partial_\omega^2\partial_{\bar\omega}^2
 f(|\omega|^2+a)
 =
 \frac{d^2}{du^2}\left[u^2f''(u+a)\right].
 \label{eq:radialderivative}
\end{equation}
Applying this identity to the $n^2$ coefficient
$f=\Phi_X$ gives
$c_a^X(u)=d^2[u^2\Phi_X''(u+a)]/du^2$.  Comparison with
$c_a^X(u)=d^2[u^2\cK_X(u+a)]/du^2$ from
\eqref{eq:masterdensity}
therefore identifies
\begin{equation}
 \cK_X(s)
 =
 \frac{d^2\Phi_X(s)}{ds^2}
 =
 \left.
 \frac12\frac{\partial^4}{\partial s^2\partial n^2}
 \log\cI_n^X(s)\right|_{n=0}.
 \label{eq:replicapotential}
\end{equation}
Only the second derivative of $\Phi_X$ enters the density covariance, so
$\Phi_X$ is determined only up to $\Phi_X(s)\mapsto\Phi_X(s)+A_X+B_Xs$.

The Hermitian input can be written explicitly as follows.  Define
$S(x)=\sin(\pi x)/(\pi x)$.
In the unit-mean-spacing normalization, the GUE, GOE, and GSE bulk pair
functions used in \cite{ChenXiaoLiuRyu2026} are
\begin{align}
 R_{2,A}^{\mathrm H}(x)
 &=1-S(x)^2,
 \label{eq:HermitianA}\\
 R_{2,AI}^{\mathrm H}(x)
 &=1-S(x)^2-S'(|x|)
   \int_{|x|}^{\infty}S(y)\dd y,
 \label{eq:HermitianAI}\\
 R_{2,AII}^{\mathrm H}(x)
 &=1-S(2x)^2+S'(2x)
   \int_0^{2x}S(y)\dd y.
 \label{eq:HermitianAII}
\end{align}
These are respectively the Hermitian partners of the non-Hermitian
classes $A$, $\AiD$, and $\AiiD$.

Let $\mathcal R_X(\zeta)$ denote the upper-half-plane analytic completion
whose boundary real part is the corresponding Hermitian pair function:
\[
 \Re\mathcal R_A(x+i0)=R_{2,A}^{\mathrm H}(x),\qquad
 \Re\mathcal R_{\AiD}(x+i0)=R_{2,AI}^{\mathrm H}(x),\qquad
 \Re\mathcal R_{\AiiD}(x+i0)=R_{2,AII}^{\mathrm H}(x),
\]
with $\mathcal R_X(\zeta)\to1$ as $|\zeta|\to\infty$ in
$\Im\zeta>0$.  Put $q=-2\pi i\zeta$, so that $\Re q>0$.  The three analytic completions
are
\begin{align}
 \mathcal R_A(\zeta)
 &=
 1+\frac{2(1-e^{-q})}{q^2},
 \label{eq:analyticRA}\\
 \mathcal R_{\AiD}(\zeta)
 &=
 1-\frac{4}{q^2}\,
 \sinh^2\!\frac q2\,\frac{\dd}{\dd q}
 \left[\frac{qE_1(q/2)}{\sinh(q/2)}\right],
 \label{eq:analyticRAI}\\
 \mathcal R_{\AiiD}(\zeta)
 &=
 1+\frac{e^{-q}}{q^2}
 \left[(1+q)\operatorname{Shi}(q)+\sinh q\right].
 \label{eq:analyticRAII}
\end{align}
Here $E_1$ is taken on its principal branch: because $\Re q>0$, no branch
cut is crossed.  The functions $E_1$ and $\operatorname{Shi}$ are defined
in \eqref{eq:Adef} and \eqref{eq:Bdef}.  Its connected part is
$\mathcal R_X^{\mathrm c}(\zeta):=\mathcal R_X(\zeta)-1$.

To derive the continuation formula, let
$\cI_{n,\mathrm H}^{X_{\mathrm H}}(x)$ denote the Hermitian compact-group
integral for the partner class: $X_{\mathrm H}=A,AI,AII$ when
$X=A,\AiD,\AiiD$, respectively.  The fixed-integer-replica identities of
\cite[Eqs.~(40), (45), and (48)]{ChenXiaoLiuRyu2026} read
\begin{equation}
 \cI_n^X(s)
 =
 \left.
 \cI_{n,\mathrm H}^{X_{\mathrm H}}(x)
 \right|_{x=is/(2\pi)}.
 \label{eq:fixednHermitianNHduality}
\end{equation}
The same work defines the analytic completion of the Hermitian pair
function by
\begin{equation}
 \mathcal R_X(\zeta)
 =
 \frac12-\frac{1}{2\pi^2}
 \lim_{n\to0}\frac{1}{n^2}
 \frac{\partial^2}{\partial\zeta^2}
 \cI_{n,\mathrm H}^{X_{\mathrm H}}(\zeta).
 \label{eq:HermitianreplicaR}
\end{equation}
Every $n\to0$ limit in this paragraph refers to the analytic continuation
selected by the Hermitian replica solution.

Differentiating \eqref{eq:fixednHermitianNHduality} twice with respect to
$s$ and using \eqref{eq:HermitianreplicaR} gives
\begin{align}
 \lim_{n\to0}\frac{1}{n^2}
 \frac{\dd^2}{\dd s^2}\cI_n^X(s)
 &=
 -\frac{1}{4\pi^2}
 \left.
 \lim_{n\to0}\frac{1}{n^2}
 \frac{\partial^2}{\partial\zeta^2}
 \cI_{n,\mathrm H}^{X_{\mathrm H}}(\zeta)
 \right|_{\zeta=is/(2\pi)}
 \notag\\
 &=
 \frac12\mathcal R_X\left(\frac{is}{2\pi}\right)-\frac14.
 \label{eq:rawreplicacoefficient}
\end{align}

It remains to pass from the replica integral to its connected logarithm.
The term linear in $n$ is fixed by the one-point bulk potential.  If
$L_0=\log|P_X(z_0;X)|^2$, then the bulk intensity $1/\pi$ implies
$\E(L_++L_-)-2\E L_0=u/2+o(1)$.  On the other hand, taking the logarithm of
\eqref{eq:fixedreplicareduction} and differentiating with respect to $n$
at $n=0$ gives
$-a/2+\left.\partial_n\log\cI_n^X(u+a)\right|_{n=0}=u/2$. Since $s=u+a$, this is equivalent to
$\left.\partial_n\log\cI_n^X(s)\right|_{n=0}=s/2$.
We may therefore write
\begin{equation}
 \cI_n^X(s)
 =
 1+\frac{ns}{2}+n^2F_X(s)+O(n^3),
 \label{eq:rawreplicaexpansion}
\end{equation}
and hence
\begin{align}
 \log\cI_n^X(s)
 &=
 \frac{ns}{2}
 +n^2\left[F_X(s)-\frac{s^2}{8}\right]
 +O(n^3)
 \notag\\
 &=
 \frac{ns}{2}+n^2\Phi_X(s)+O(n^3),
 \label{eq:connectedreplicaexpansion}
\end{align}
where
\begin{equation}
 \Phi_X(s)=F_X(s)-\frac{s^2}{8}.
 \label{eq:defreplicapotential}
\end{equation}
Equation \eqref{eq:rawreplicacoefficient} states that
$F_X''(s)=\tfrac12\mathcal R_X(is/2\pi)-\tfrac14$.
Consequently,
\begin{equation}
 \cK_X(s)=\Phi_X''(s)
 =
 \frac12\left[
 \mathcal R_X\left(\frac{is}{2\pi}\right)-1
 \right].
 \label{eq:Hermitiancontinuation}
\end{equation}
Thus the subtraction of $s^2/8$ generated by taking
$\log\cI_n^X$ changes the raw term $-1/4$ in
\eqref{eq:rawreplicacoefficient} into $-1/2$.  This is precisely the
subtraction of the disconnected background $1$ from the pair function.

Equivalently, for any fixed $s_*>0$,
\begin{equation}
 \Phi_X(s)
 =
 A_X+B_Xs+
 \frac12\int_{s_*}^{s}(s-t)
 \left[
  \mathcal R_X\left(\frac{it}{2\pi}\right)-1
 \right]\dd t.
 \label{eq:replicapotentialintegrated}
\end{equation}
The constants $A_X$ and $B_X$ do not affect
$\Phi_X''=\cK_X$.  Substitution of
\eqref{eq:analyticRA}--\eqref{eq:analyticRAII} into
\eqref{eq:Hermitiancontinuation} gives the three kernels in
\eqref{eq:threeK}.

Combining the radial source identity \eqref{eq:radialderivative} with
$\Phi_X''=\cK_X$ from \eqref{eq:Hermitiancontinuation} gives the predicted
density \eqref{eq:masterdensity}.  The factor $e^{-na/2}$ in
\eqref{eq:fixedreplicareduction} contributes only to the coefficient
linear in $n$ and hence does not enter the connected density covariance.

The continuation used here is not an arbitrary interpolation of the
positive-integer replica values.  Its Hermitian boundary data are the exact
GUE, GOE, and GSE pair functions.  In replica language these functions can
be recovered from integrable hierarchies: the Toda/Painlev\'e structure in
class $A$ \cite{Kanzieper2002,SplittorffVerbaarschot2003}, the Pfaff--KP
hierarchy for class $AI$ \cite{VidalKanzieper2013}, and its dual counterpart
for class $AII$ as summarized in \cite{ChenXiaoLiuRyu2026}.  This integrable structure selects
and independently checks the continuation employed above.  It does not,
however, supply the uniform control near $n=0$ needed to interchange the
replica limit, the large-$N$ limit, and the source derivatives in
\cref{prop:fixedreplica}.

Finally, the factor $\chi_0=1-|z_0|^2$ enters the parameter scale through
$a=v^2\chi_0$, whereas the microscopic mean spectral intensity remains
$1/\pi$.  At the spectral edge, $\chi_0\downarrow0$ and
$\operatorname{Hess}_{\mathrm{rad}}S_0=4\gamma_X\chi_0\downarrow0$,
so the bulk radial expansion ceases to be uniform and must be replaced by
an edge scaling analysis.

\section{Contact normalization and small-parameter eigenvalue motion}
\label{sec:self}

\subsection{What the static derivative forgets}

Knowing $F''(u)$ determines $F$ only up to $b_0+b_1u$, so that for
$F(u)=u^2\cK_X(u)$ these ambiguities become $b_0/u^2+b_1/u$ in the kernel,
disappearing from a static distinct-level correlation but becoming visible
after the shift $u\mapsto u+a$.

It is necessary to specify boundary conditions before interpreting this
observation dynamically.  The kernels \eqref{eq:threeK} have
\begin{equation}
 F(0)=0,\qquad F'(0)=1,\qquad F'(\infty)=0.
 \label{eq:contactconditions}
\end{equation}
The condition $F(0)=0$ ensures regularity at the origin, whereas
$F'(0)=1$ fixes the unit weight of the equal-eigenvalue contact term in
\eqref{eq:staticmeasure}.  They
give the explicit static reconstruction
\begin{equation}
 F(u)=u+\int_0^u(u-t)c_0^X(t)\,dt.
 \label{eq:staticreconstruction}
\end{equation}
Thus static data supplemented by intensity and contact conventions can also
fix the integration constants; dynamics is not their only possible source.
The final condition gives $\int_0^\infty c_0^X(t)dt=-1$.

For positive $a$, the singular mode produces
\begin{equation}
 q_a(u):=\frac{d^2}{du^2}\frac{u^2}{u+a}
 =\frac{2a^2}{(u+a)^3},\qquad \int_0^\infty q_a(u)\,du=1.
 \label{eq:selfpeak}
\end{equation}
Since $\int_{\C}f(|w|^2)d^2w=\pi\int_0^\infty f(u)du$,
$\pi^{-2}q_a(|w|^2)d^2w$ converges to
$\pi^{-1}\delta^{(2)}(w)$, precisely the contact mass in
\eqref{eq:staticmeasure}.  The regular remainder carries total radial weight
$-1$.  This explains how a zero-total-mass positive-time covariance tends to
a static contact term plus a distinct-level correlation hole.

The normalization $\int_0^\infty c_a^X(u)\,du=0$ does not by itself
exclude an additional term $b_0/s^2$ in $\cK_X(s)$, because the
corresponding contribution after $s=u+a$ also integrates to zero.  This
ambiguity is removed by the regularity condition $F(0)=0$ in
\eqref{eq:contactconditions}, which excludes an $s^{-2}$ singularity of
$\cK_X(s)$ as $s\downarrow0$.  More quantitatively, the additional radial
term
\[
 b_0\left[-\frac{4a}{(u+a)^3}+\frac{6a^2}{(u+a)^4}\right]
\]
has zero integral over $u\geq0$, whereas the integral of its absolute value
is $16|b_0|/(27a)$.  Its positive and negative parts therefore each have
weight $8|b_0|/(27a)$, which diverges as $a\downarrow0$.  Hence this term
cannot converge as a finite signed measure to the contact-plus-distinct-level
limit described above.

\subsection{An exact velocity calculation in \texorpdfstring{classes $A$ and $\AiD$}{classes A and AI dagger}}

Let $Xr_i=\lambda_i r_i$ and $\ell_i^\dagger X=\lambda_i\ell_i^\dagger$,
with $\ell_i^\dagger r_i=1$.  For a simple eigenvalue and an independent
unit-strength matrix $Y$ in the same ensemble,
\[
 \left.\frac{d}{d\varepsilon}\lambda_i(X+\varepsilon Y)
 \right|_{\varepsilon=0}=\ell_i^\dagger Yr_i,
 \qquad O_{ii}=(\ell_i^\dagger\ell_i)(r_i^\dagger r_i).
\]
Conditional on $X$, the real and imaginary parts of $\dot\lambda_i$ are
independent centered real Gaussian variables with equal variance.  Thus
$\E(\dot\lambda_i\mid X)=\E(\dot\lambda_i^2\mid X)=0$, while its
conditional absolute second moment is
\begin{equation}
 \E(|\dot\lambda_i|^2\mid X)=
 \begin{cases}
 O_{ii}/N,&A,\\
 2O_{ii}/N,&\AiD.
 \end{cases}
 \label{eq:velocitytensors}
\end{equation}  
Indeed, in the symmetric case
$\ell_i^\dagger=r_i^T/(r_i^Tr_i)$, and the two contractions in
\eqref{eq:AIcov} each give
$(r_i^\dagger r_i)^2/(N|r_i^Tr_i|^2)$.

At a bulk point $z_0$, independent results
\cite{BourgadeDubach2020,Fyodorov2018,AkemannFyodorovSavin2026}
give the limiting distribution of the rescaled diagonal overlap observed
at an eigenvalue near $z_0$:
\begin{equation}
 m_A=\frac{O_{ii}}{N\chi_0},\qquad
 m_{\AiD}=\frac{2O_{ii}}{N\chi_0},\qquad
 p(m)=m^{-3}e^{-1/m}\1_{m>0}.
 \label{eq:knownoverlaps}
\end{equation}
More explicitly, if $m_i$ denotes the corresponding rescaled overlap,
then for any bounded test function $g$,
\[
 \frac{
 \E\sum_i\delta^{(2)}(z_0-\lambda_i)\,g(m_i)}
 {\E\sum_i\delta^{(2)}(z_0-\lambda_i)}
 \longrightarrow
 \int_0^\infty g(m)p(m)\,dm .
\]
The ratio on the left is the eigenvalue-resolved average of $g(m_i)$ at
$z_0$, normalized by the mean local eigenvalue density.  Equivalently,
one selects an eigenvalue near $z_0$ and records its rescaled overlap.
In probability terminology, this eigenvalue-conditioned distribution is
called the one-point Palm distribution.  In the complex-symmetric reference the spectral disk
has radius $\sqrt2$; putting its $r=\sqrt2|z_0|$ gives the scale $N\chi_0/2$ in
\eqref{eq:knownoverlaps}.  A scalar rescaling of a matrix does not change its
eigenvector overlaps.

For the correlated matrix pair $(X_1,X_\alpha)$ defined in
\eqref{eq:interpolation}, the leading microscopic eigenvalue displacement is
$p=\sqrt N\,\delta\lambda=v\dot\lambda$. Its conditional
variance is $am$ in either class.  Therefore
\begin{equation}
 \int_0^\infty\frac{e^{-|p|^2/(am)}}{\pi am}
 \frac{e^{-1/m}}{m^3}\,dm
 =\frac{2a^2}{\pi(a+|p|^2)^3}.
 \label{eq:velocitymixture}
\end{equation}
This calculation is exact for the linearized displacement and uses the
limiting bulk law of the rescaled overlap.  Its application to the full
displacement requires
the additional small-parameter regime.  The conjectured connected two-time density correlation obtained by replica
continuation has the matching short-time scaling limit
\begin{equation}
 a\,c_a^X(ax)\longrightarrow\frac{2}{(1+x)^3},
 \qquad a\downarrow0,
 \label{eq:smalltimescaling}
\end{equation}
for every fixed $x\geq0$.  This follows by substituting the explicit
small-$s$ expansions listed above into \eqref{eq:masterdensity}: the common
leading term $\cK_X(s)=s^{-1}$ produces the displayed limit, while the
remaining terms vanish after multiplication by $a$. In class $\AiD$,
the equality of the conditional displacement variances is the explicit
cancellation $2v^2O_{ii}/N=v^2\chi_0[2O_{ii}/(N\chi_0)]=am_{\AiD}$:
the factor $2$ in the matrix-noise contraction
\eqref{eq:velocitytensors} is compensated by the factor $1/2$ in the
natural bulk scale of $O_{ii}$ in \eqref{eq:knownoverlaps}.

The comparison involves two different iterated limits.  Finite-$N$
perturbation theory first takes $v\downarrow0$ at fixed $N$, in the order
$\lim_{N\to\infty}\lim_{v\downarrow0}$,
thereby replacing the microscopic displacement by its linearization
$p_{i,N}(v)=v\dot\lambda_{i,N}+o(v)$, and only afterwards uses the
large-$N$ eigenvalue-conditioned limit of the rescaled overlap.  In contrast,
the replica
calculation first takes the microscopic bulk limit at fixed
$a=v^2\chi_0>0$ and only then lets $a\downarrow0$:
\[
 \lim_{a\downarrow0}\lim_{N\to\infty}
 a\,c_{a,N}^X(ax)
 =
 \frac{2}{(1+x)^3}.
\]
Identifying these two limits rigorously would require
$p_{i,N}(v)=v\dot\lambda_{i,N}+o_{\mathbb P}(v)$
to hold uniformly in the large-$N$ Palm regime relevant to
$|p_{i,N}(v)|^2=O(a)$, including configurations with small eigenvalue gaps
and large eigenvector overlaps.  The agreement with
\eqref{eq:smalltimescaling} therefore provides an independent check of the
normalization of $a$ and of the self peak.

Here ``finite $a$'' means that the microscopic parameter
$a=v^2\chi_0>0$ is held fixed as $N\to\infty$; for the stationary
Ornstein--Uhlenbeck process this corresponds to a physical time difference
of order $N^{-1}$.  In this regime,
\[
 X_\alpha-X_1
 =
 \frac{v}{\sqrt N}X_2-\frac{v^2}{2N}X_1+O(N^{-3/2}),
\]
so the matrix increment and the typical bulk eigenvalue spacing are both
of order $N^{-1/2}$.  Consequently, mixing with neighboring eigenvalues
cannot in general be neglected.  For example, the perturbation series for
a simple eigenvalue begins as
\[
 \lambda_i(X+\varepsilon Y)-\lambda_i(X)
 =
 \varepsilon\,\ell_i^\dagger Yr_i
 +\varepsilon^2
 \sum_{j\ne i}
 \frac{(\ell_i^\dagger Yr_j)(\ell_j^\dagger Yr_i)}
      {\lambda_i-\lambda_j}
 +O(\varepsilon^3),
\]
and the small denominators $\lambda_i-\lambda_j=O(N^{-1/2})$ prevent this
expansion from being uniformly reducible to its first term at fixed
microscopic $a$.

There is also a distinction between endpoint and trajectory observables.
The two-time point-process intensity calculated in this paper is
\[
 \mathcal G_a(z,w)
 =
 \E\sum_{i,j}
 \delta^{(2)}(z-\lambda_i(0))
 \delta^{(2)}(w-\lambda_j(a))
 =
 \E[\rho_0(z)\rho_a(w)].
\]
It depends only on the joint distribution of the two endpoint spectra and
is invariant under independent permutations of their eigenvalue labels.
By contrast, once the full matrix evolution $X(t)$ has been specified at
all intermediate times and its eigenvalues have been followed continuously one may define the pathwise self intensity
\[
 \mathcal S_a(z,w)
 =
 \E\sum_i
 \delta^{(2)}(z-\lambda_i(0))
 \delta^{(2)}(w-\lambda_i(a)).
\]
The endpoint quantity $\mathcal G_a$ does not determine
$\mathcal S_a$: the latter additionally requires matching the endpoint eigenvalues by
continuous tracking along the chosen matrix evolution.

The decomposition used here follows instead from the singular part
$\cK_X(s)=s^{-1}+\cK_X^{\rm reg}(s)$.
It gives the exact algebraic identity
\begin{equation}
 c_a^X(u)
 =
 q_a(u)+r_a^X(u),
 \qquad
 q_a(u)=\frac{2a^2}{(u+a)^3},
 \qquad
 r_a^X(u)=
 \frac{\dd^2}{\dd u^2}
 \left[u^2\cK_X^{\rm reg}(u+a)\right].
 \label{eq:selfregdecomposition}
\end{equation}
The first term is positive, has unit radial mass
$\int_0^\infty q_a(u)\dd u=1$, and concentrates on $u=O(a)$.  Moreover,
\[
 a\,q_a(ax)=\frac{2}{(1+x)^3},
 \qquad
 a\,r_a^X(ax)\longrightarrow0
 \quad (a\downarrow0),
\]
so $q_a$ is precisely the leading small-parameter self peak identified by
the velocity calculation.  At fixed positive $a$,
\eqref{eq:selfregdecomposition} remains an exact decomposition of the
endpoint covariance, whereas the pathwise quantity $\mathcal S_a$ is a
different, trajectory-resolved observable.  Determining
$\mathcal S_a$ for finite $a$ would require a multitime calculation that
retains both the eigenvalue matching along $X(t)$ and the corresponding
eigenvector overlaps.

\section{The intrinsic Kramers-plane diffusivity}
\label{sec:kramers}

\subsection{A plane invariant and a scalar doublet velocity}

Let $R=(r_1,r_2)$ span a semisimple right eigenspace of a self-dual
matrix, so that the symplectic form restricted to this plane is
$S=R^T\mathbb JR=\mathfrak sJ_2$, where
$J_2=\left(\begin{smallmatrix}0&1\\-1&0\end{smallmatrix}\right)$ and
$\mathfrak s\ne0$, its nondegeneracy following because distinct
eigenspaces are symplectically orthogonal and the total form is
nondegenerate.  Define
\begin{equation}
 \boxed{\cO_K(R)=\frac{\det(R^\dagger R)}{|\Pf(R^T\mathbb JR)|^2}
 =\frac{(r_1^\dagger r_1)(r_2^\dagger r_2)-|r_1^\dagger r_2|^2}
 {|r_2^T\mathbb Jr_1|^2}.}
 \label{eq:OK}
\end{equation}
Under $R\mapsto RB$, with $B\in GL(2,\C)$, numerator and denominator are
both multiplied by $|\det B|^2$.  Thus the overlap belongs to the plane,
not to a selected eigenvector inside it.  Moreover, $\cO_K\ge1$.  Indeed,
$\mathbb J\bar r_1$ is orthogonal to $r_1$ and has norm $\|r_1\|$; applying
Cauchy--Schwarz to the component of $r_2$ orthogonal to $r_1$ gives
$|r_2^T\mathbb Jr_1|^2\le\det(R^\dagger R)$.  Equality holds when the plane
admits an orthonormal Kramers pair, in particular for a normal self-dual
matrix.  Thus $\cO_K$ is literally a condition number: it measures the
departure of the eigenspace from a normal Kramers plane.

\begin{proposition}[Kramers-plane perturbation and covariance]
\label{prop:kramers}
For a self-dual perturbation $W$, satisfying
$W=\mathbb JW^T\mathbb J^{-1}$, both members of the doublet have the same
first-order shift
\begin{equation}
 \dot\lambda
 =\frac{r_2^T\mathbb JWr_1}{r_2^T\mathbb Jr_1}.
 \label{eq:Kvelocity}
\end{equation}
If $W$ is independent of $X$ and is drawn from the unit-disc-normalized
Gaussian class-$\AiiD$ ensemble defined by
\eqref{eq:selfdual}--\eqref{eq:AIIcov}, then
\[
 \E W_{ij}\overline{W_{kl}}
 =
 \frac1{2N}
 \left(\delta_{ik}\delta_{jl}+\mathbb J_{il}\mathbb J_{jk}\right),
 \qquad
 \E W_{ij}W_{kl}=0,
\]
and
\begin{equation}
 \E\bigl(|\dot\lambda|^2\mid R\bigr)
 =\frac{\cO_K(R)}{2N},
 \qquad
 \E\bigl(\dot\lambda^2\mid R\bigr)=0.
 \label{eq:Kvariance}
\end{equation}
The spectral projector $\Pi=RS^{-1}R^T\mathbb J$ satisfies
\begin{equation}
 \cO_K=\frac12\Tr(\Pi\Pi^\dagger),\qquad
 \dot\lambda=\frac12\Tr(\Pi W).
 \label{eq:projectoroverlap}
\end{equation}
\end{proposition}

\begin{proof}
Self-duality gives $X^T\mathbb J=\mathbb JX$, so $S^{-1}R^T\mathbb J$
is a left dual basis.  The matrix $R^T\mathbb JWR$ is antisymmetric and hence equals $tJ_2$ for a
scalar $t$, its product with $S^{-1}$ being $(t/\mathfrak s)\1_2$, which
proves the scalar shift together with its independence of the chosen
basis.

For any vectors $a,b$, \eqref{eq:AIIcov} gives
\[
 \E|a^TWb|^2=\frac1{2N}
 [\|a\|^2\|b\|^2+(a^T\mathbb J\bar b)(b^T\mathbb J\bar a)].
\]
Set $a=-\mathbb Jr_2$ and $b=r_1$.  Since
$\mathbb J^\dagger\mathbb J=\1$, one has
$\|a\|^2=r_2^\dagger r_2$, $\|b\|^2=r_1^\dagger r_1$,
$a^T\mathbb J\bar b=-r_1^\dagger r_2$, and
$b^T\mathbb J\bar a=r_2^\dagger r_1$.
Therefore
\begin{align}
 \E\left(
 \left|r_2^T\mathbb JWr_1\right|^2\mid R
 \right)
 &=
 \frac1{2N}
 \left[
  (r_1^\dagger r_1)(r_2^\dagger r_2)
  -|r_2^\dagger r_1|^2
 \right]=
 \frac{\det(R^\dagger R)}{2N}.
 \label{eq:Knumeratorvariance}
\end{align}
Moreover, $\E W_{ij}W_{kl}=0$ in \eqref{eq:AIIcov}, so
$\E[(r_2^T\mathbb JWr_1)^2\mid R]=0$.  Dividing by
$|r_2^T\mathbb Jr_1|^2=|\Pf(R^T\mathbb JR)|^2=|\mathfrak s|^2$ gives
$\E(|\dot\lambda|^2\mid R)=\cO_K(R)/(2N)$ and
$\E(\dot\lambda^2\mid R)=0$,
which proves \eqref{eq:Kvariance}.

It remains to verify the projector identities.  Put $G=R^\dagger R$ and
$S=R^T\mathbb JR=\mathfrak sJ_2$.
From $\Pi=RS^{-1}R^T\mathbb J$ and
$\mathbb J^\dagger\mathbb J=\1$,
\begin{align}
 \Tr(\Pi\Pi^\dagger)
 &=
 \Tr\left[
  RS^{-1}R^T\bar R(S^{-1})^\dagger R^\dagger
 \right]=
 \Tr\left[
  GS^{-1}G^T(S^{-1})^\dagger
 \right].
 \label{eq:projectorGramtrace}
\end{align}
For every $2\times2$ matrix $G$, $GJ_2G^T=(\det G)J_2$.  Since
$S^{-1}=-J_2/\mathfrak s$, $(S^{-1})^\dagger=J_2/\bar{\mathfrak s}$, and
$J_2^2=-\1_2$, equation \eqref{eq:projectorGramtrace} becomes
\begin{align}
 \Tr(\Pi\Pi^\dagger)
 &=
 -\frac1{|\mathfrak s|^2}\Tr\left(GJ_2G^TJ_2\right)
 \notag\\
 &=
 -\frac{\det G}{|\mathfrak s|^2}\Tr(J_2^2)
 =
 \frac{2\det G}{|\mathfrak s|^2}
 =
 2\cO_K(R).
\end{align}
This proves the first identity in \eqref{eq:projectoroverlap}.

Finally, write $R^T\mathbb JWR=tJ_2$ and use cyclicity of the trace, with
$S=\mathfrak sJ_2$:
\begin{align}
 \Tr(\Pi W)
 &=
 \Tr\left(S^{-1}R^T\mathbb JWR\right)
 =
 \Tr\left[(\mathfrak sJ_2)^{-1}(tJ_2)\right]
 \notag\\
 &=
 \Tr\left(\frac t{\mathfrak s}\1_2\right)
 =
 \frac{2t}{\mathfrak s}
 =
 2\dot\lambda.
\end{align}
Hence $\dot\lambda=\frac12\Tr(\Pi W)$, proving the second identity in
\eqref{eq:projectoroverlap}.
\end{proof}

These formulas also give a direct stochastic interpretation.  In stationary
Ornstein--Uhlenbeck dynamics, with the Brownian covariance normalized as in
\eqref{eq:AIIcov}, a locally followed doublet satisfies
\begin{equation}
 d\lambda_i=-\tfrac12\lambda_i\,dt+dM_i,\qquad
 d\langle M_i,\bar M_j\rangle_t
 =\frac1{4N}\Tr(\Pi_i\Pi_j^\dagger)\,dt.
 \label{eq:Kito}
\end{equation}
Here $M_i$ is the zero-mean stochastic part of the eigenvalue motion, and
$\langle\cdot,\cdot\rangle_t$ denotes quadratic covariation.  In
physicists' notation, conditional on the current matrix $X_t$, the bracket
may be read as
$d\langle M_i,\bar M_j\rangle_t
=\E[dM_i\,\overline{dM_j}\mid X_t]$.
For $i=j$, the coefficient of $dt$ is the instantaneous mean-square
displacement rate; for $i\ne j$, it gives the instantaneous covariance
between the motions of the two eigenvalue doublets.

Writing $d\mathcal B_t$ for the Brownian part of the matrix increment, its
unconjugated entrywise quadratic covariations vanish,
$d\langle(\mathcal B_t)_{ab},(\mathcal B_t)_{cd}\rangle=0$, while
\[
 d\langle(\mathcal B_t)_{ab},
 \overline{(\mathcal B_t)_{cd}}\rangle
 =
 \frac1{2N}
 \left(\delta_{ac}\delta_{bd}
 +\mathbb J_{ad}\mathbb J_{bc}\right)\dd t.
\]
Away from eigenvalue collisions, $\lambda_i$ is locally a holomorphic
function of the independent complex matrix entries.  Its second-order
holomorphic It\^o term is therefore
\[
 \frac12\sum_{a,b,c,d}
 \frac{\partial^2\lambda_i}
 {\partial X_{ab}\partial X_{cd}}\,
 d\langle(\mathcal B_t)_{ab},(\mathcal B_t)_{cd}\rangle
 =0.
\]
Thus the martingale increment is $dM_i=\tfrac12\Tr(\Pi_i\,d\mathcal B_t)$.
Its unconjugated quadratic covariation vanishes, so
$d\langle M_i,M_j\rangle_t=0$.
Taking $j=i$ in \eqref{eq:Kito} and using
\eqref{eq:projectoroverlap} gives
\begin{equation}
 d\langle M_i,\bar M_i\rangle_t
 =
 \frac1{4N}\Tr(\Pi_i\Pi_i^\dagger)\dd t
 =
 \frac{\cO_{K,i}}{2N}\dd t.
 \label{eq:Kdiagonalbracket}
\end{equation}
This is the infinitesimal form of the conditional variance
\eqref{eq:Kvariance}.
Both contractions in \eqref{eq:AIIcov} contribute equally here.  The direct
term $\delta_{ac}\delta_{bd}$ alone would give
$\Tr(\Pi_i\Pi_i^\dagger)/(8N)$, and the exchange term
$\mathbb J_{ad}\mathbb J_{bc}$ doubles it, because the spectral projector
of a self-dual matrix is itself self-dual.  This is the same doubling that
produces the factor $2$ in the complex-symmetric velocity variance
\eqref{eq:velocitytensors}.

\subsection{What follows for the overlap distribution}

Define the dimensionless Kramers-plane overlap
$m_K=\cO_K/(2N\chi_0)$.  Then the linearized microscopic displacement has
conditional variance
\begin{equation}
 \E(|p|^2\mid m_K)=am_K.
 \label{eq:universalvelocity}
\end{equation}
This is a consequence of the normalization chosen for $m_K$ and
\cref{prop:kramers}.  The independent derivation of the coefficient of $a$
is \eqref{eq:masscalculation}; \eqref{eq:universalvelocity} alone cannot
check the mean of $m_K$.

This conditional argument leads to the following concrete conjecture.
\begin{conjecture}[Kramers-plane overlap law]
Suppose the rescaled Kramers-plane overlap $m_K$ has limiting law $\mu$,
and its limiting linearized velocity mixture agrees with the
replica-derived short-time limit of the connected density correlation in
\eqref{eq:smalltimescaling}.  Together with \eqref{eq:Kvariance}, this
requires  
\begin{equation}
 \int_0^\infty m^{-1}e^{-x/m}\,\mu(dm)=\frac{2}{(1+x)^3},\qquad x\ge0.
 \label{eq:Klaplace}
\end{equation}
With $t=1/m$, this is a Laplace transform.  Its unique inverse is
\begin{equation}
 \mu(dm)=m^{-3}e^{-1/m}\,dm.
 \label{eq:Kprediction}
\end{equation}
\end{conjecture}
Thus the conjectured overlap law follows if the short-time self
contribution to the spectral correlation is identified with the
distribution generated by first-order eigenvalue motion.  An independent
finite-$N$ derivation of the joint eigenvalue--overlap density remains
open.  The same argument predicts the conditional mean
$\E(\cO_K\mid z_0)\sim2N\chi_0$.  Convergence of the rescaled overlap
distribution to \eqref{eq:Kprediction} would not by itself prove this
mean asymptotic, because rare exceptionally large overlaps could still
make a non-negligible contribution; indeed, the predicted density decays
only as $p(m)\sim m^{-3}$.  The simulations in \cref{sec:numerics}
therefore test the distribution through inverse moments, which are finite
and much less sensitive to such rare large-overlap events.

\section{Static checks and the Calogero representation}

The exact finite-$N$ joint density of \cite{XiaoChenLiuRyu2026} provides an
independent description of the static ensembles.  To state it without
mixing normalizations, use that paper's coordinates
$w_j=\sqrt{kN}\,z_j$, where $k=\beta/2$.  Its Gaussian weight is
$e^{-\Tr H^\dagger H/\kappa}$, with $\kappa=1$ for $A,\AiD$ and $\kappa=2$
for $\AiiD$, and its bulk density in the $w$ plane is $1/(\pi k)$.
In those variables,
\begin{equation}
 \rho_{N,X}(w_1,\ldots,w_N)
 =Z_{N,\beta}^{-1}|\Delta(w)|^\beta
 \Psi_k(w,i\bar w),\qquad
 \Delta(w)=\prod_{i<j}(w_j-w_i).
 \label{eq:calogerodensity}
\end{equation}
Here $\Psi_k$ is the asymptotically free scattering solution of the rational
Calogero Hamiltonian
\[
 H_k=-\sum_i\partial_{w_i}^2
 +\sum_{i\ne j}\frac{k(k-1)}{(w_i-w_j)^2},
\]
analytically evaluated at the indicated coordinates and momenta.  For $k=2$
it is a Baker--Akhiezer function; for $k=1/2$ the cited construction gives a
noncompact orbital-integral representation.  The Gaussian confinement is
contained in the scattering solution at $p=i\bar w$; an additional Gaussian
factor should not be inserted into \eqref{eq:calogerodensity}.

The exact finite-$N$ identity and the large-distance asymptotics derived in
\cite{XiaoChenLiuRyu2026} have different logical status.  The finite-$N$
joint-density representation is exact, whereas Xiao {\it et al.}
obtain the large-distance tail using a Wentzel--Kramers--Brillouin (WKB)
expansion together with a random-phase approximation for the structure
factor.  Their normalized connected correlation function
$h_X=\rho_{2,X}/\rho_{1,X}^2-1$ then has the asymptotic form
\[
 h_X(w_1,w_2)
 \sim\frac{4k(k-1)}{|w_1-w_2|^6}.
\]
 Dividing the Calogero coordinates $w_j$ by $\sqrt k$ changes their bulk
density from $1/(\pi k)$ to our normalization $1/\pi$.  The corresponding
spectral separation is therefore
$r=|w_1-w_2|/\sqrt k=\sqrt u$.
Recall from \eqref{eq:staticmeasure} that $c_0^X(u)$ denotes the
dimensionless distinct-level part of the equal-parameter connected
covariance, with the diagonal self-correlation removed.  Equivalently,
\begin{equation}
 c_0^X(u)=\pi^2R_2^X(u)-1,
 \label{eq:c0pairrelation}
\end{equation}
where $R_2^X$ is the distinct-level pair intensity at bulk density
$1/\pi$.  In this normalization, the large-distance predictions of
\cite{XiaoChenLiuRyu2026} become
\begin{equation}
 c_0^{\AiD}(r^2)\sim-\frac8{r^6},
 \qquad
 c_0^{\AiiD}(r^2)\sim\frac1{r^6}.
 \label{eq:Calogerotails}
\end{equation}
If one instead denotes the unnormalized connected pair intensity by
$\rho_2^{(c)}$, then $c_0=\pi^2\rho_2^{(c)}$ in our microscopic coordinates.

The kernel check is especially transparent.  If
$\cK_X(u)=d_Xu^{-2}+e_Xu^{-3}+O(u^{-4})$, then
\[
 \frac{d^2}{du^2}[u^2\cK_X(u)]=2e_Xu^{-3}+O(u^{-4}).
\]
The values $e_{\AiD}=-4$ and $e_{\AiiD}=1/2$ give exactly
\eqref{eq:Calogerotails}, so that the test probes the first subleading
coefficient including its sign, whereas in class $A$ one has $e_A=0$ and
the static correlation is instead $-e^{-u}$.  Thus the kernel asymptotics reproduce the independently known static tails
in all three symmetry classes, including their signs and normalizations.

There is also a direct static formula check, independent of a tail
expansion.  Since $u^2\cK_{\AiD}(u)=-2\mathfrak A(u)$ and
$u^2\cK_{\AiiD}(u)=\mathfrak B(u)$, \eqref{eq:static} reads
\[
 \pi^2R_2^{\AiD}(u)=1-2\mathfrak A''(u),\qquad
 \pi^2R_2^{\AiiD}(u)=1+\mathfrak B''(u).
\]
With \eqref{eq:Adef}--\eqref{eq:Bdef}, these are the two explicit functions
in Eq.~(16) of \cite{ChenXiaoLiuRyu2026}, in the same unfolded coordinates.

\section{Counting statistics and the geometry of a disk}

\subsection{The disk-overlap Abel transform and its contact term}

Let $D_R=\{z\in\C:|z|\leq R\}$ and
$\mathcal N_\sigma(D_R)=\int_{D_R}\rho_\sigma(z)\dd^2z$.
For a translation-invariant microscopic covariance measure $C_a$, Fubini's
formula gives
\begin{align}
 \Cov\bigl(\mathcal N_0(D_R),\mathcal N_a(D_R)\bigr)
 &=
 \int_{\C}
 \left[
  \int_{\C}
  \1_{D_R}(z)\1_{D_R}(z-w)\dd^2z
 \right]C_a(dw)
 \notag\\
 &=
 \int_{\C}
 \bigl|D_R\cap(D_R-w)\bigr|\,C_a(dw),
 \label{eq:lenscovariance}
\end{align}
where $|\cdot|$ denotes planar area.

At microscopic intensity $1/\pi$,
$\E\mathcal N_\sigma(D_R)=|D_R|/\pi=R^2$.  We therefore set $y=R^2$ and
$u=|w|^2$,
and define the disk-overlap area
\[
 G_y(u)
 :=
 \bigl|D_{\sqrt y}\cap(D_{\sqrt y}-w)\bigr|,
 \qquad |w|=\sqrt u.
\]

Set $R=\sqrt y$ and $d=\sqrt u$.  When $d\leq2R$, the intersection of
the two disks consists of two congruent circular segments.  Each segment
is a sector of half-angle $\arccos(d/(2R))$ with the corresponding
isosceles triangle removed.  Summing the two sector areas and subtracting
the two triangle areas gives
\begin{equation}
 G_y(u)
 =
 2y\arccos\left(\frac{\sqrt u}{2\sqrt y}\right)
 -\frac{\sqrt u}{2}\sqrt{4y-u},
 \qquad 0\leq u\leq4y.
 \label{eq:lensarea}
\end{equation}
For $u\geq4y$, the disks do not overlap and $G_y(u)=0$. In particular,
$G_y(0)=\pi y$ and $G_y(4y)=0$; for $0<u<4y$,
\begin{equation}
 G_y'(u)=-\frac{\sqrt{4y-u}}{2\sqrt u},
 \qquad
 G_y''(u)=\frac{y}{u^{3/2}\sqrt{4y-u}}.
 \label{eq:lensderivatives}
\end{equation}

For a function $f$ for which the integral converges, define the
\emph{disk Abel transform}
\begin{equation}
 (\mathcal A_{\mathrm{disk}}f)(y)
 :=
 \frac{y}{\pi}\int_0^{4y}
 \frac{\sqrt u\,f(u)}{\sqrt{4y-u}}\dd u.
 \label{eq:diskAbeltransform}
\end{equation}
It is a weighted form of the classical Abel integral
$g\mapsto\int_0^L g(u)(L-u)^{-1/2}\dd u$; the square-root kernel in
\eqref{eq:diskAbeltransform} is the reason for the terminology.

\begin{proposition}[Disk covariance as an Abel transform]
\label{prop:diskAbel}
Let $a>0$ and suppose that the connected microscopic covariance measure is
\[
 C_a(dw)=\frac1{\pi^2}F_a''(|w|^2)\dd^2w,
 \qquad
 F_a(u)=u^2f_a(u),
 \qquad
 f_a(u)=\cK_X(u+a),
\]
where $\cK_X$ is one of the kernels in \eqref{eq:threeK}.  Then
\begin{equation}
 \Cov\bigl(\mathcal N_0(D_R),\mathcal N_a(D_R)\bigr)
 =
 (\mathcal A_{\mathrm{disk}}f_a)(y)
 =
 \frac{y}{\pi}\int_0^{4y}
 \frac{\sqrt u\,\cK_X(u+a)}{\sqrt{4y-u}}\dd u,
 \label{eq:diskAbelcovariance}
\end{equation}
which is \eqref{eq:countcov}.

At equal parameters, let the full covariance measure be
\[
 C_0^X(dw)
 =
 \frac1\pi\delta^{(2)}(w)\dd^2w
 +\frac1{\pi^2}F_0''(|w|^2)\dd^2w,
 \qquad
 F_0(u)=u^2\cK_X(u).
\]
Then
\begin{equation}
 \Var N(D_R)
 =
 (\mathcal A_{\mathrm{disk}}\cK_X)(y),
 \label{eq:staticdiskAbel}
\end{equation}
which is \eqref{eq:numbervariance}.
\end{proposition}

\begin{proof}
For $a>0$, substitution of the absolutely continuous covariance measure
into \eqref{eq:lenscovariance}, followed by the polar change of variables
$u=|w|^2$, gives
\begin{equation}
 \Cov\bigl(\mathcal N_0(D_R),\mathcal N_a(D_R)\bigr)
 =
 \frac1\pi\int_0^{4y}G_y(u)F_a''(u)\dd u.
 \label{eq:radialdiskcovariance}
\end{equation}
Integrating twice by parts on
$[\varepsilon,4y-\varepsilon]$ yields
\begin{align}
 \frac1\pi\int_0^{4y}G_yF_a''\dd u
 &=
 \frac1\pi
 \left[G_yF_a'-G_y'F_a\right]_{0}^{4y}
 +\frac1\pi\int_0^{4y}G_y''F_a\dd u.
 \label{eq:diskIBP}
\end{align}
At $u=4y$, both $G_y(u)$ and $G_y'(u)$ vanish.  Since $a>0$,
the kernel is regular at $u+a=a$, and hence
$F_a(u)=u^2\cK_X(a)+O(u^3)$ and
$F_a'(u)=2u\cK_X(a)+O(u^2)$ as $u\downarrow0$.
Together with $G_y(0)=\pi y$ and
$G_y'(u)=O(u^{-1/2})$, both lower-end products
$G_yF_a'$ and $G_y'F_a$ tend to zero.
Thus the boundary term in \eqref{eq:diskIBP} vanishes.  Using
\eqref{eq:lensderivatives} and $F_a(u)=u^2f_a(u)$ in the remaining integral,
we obtain
\begin{align}
 \frac1\pi\int_0^{4y}G_y''(u)F_a(u)\dd u
 &=
 \frac{y}{\pi}\int_0^{4y}
 \frac{u^2f_a(u)}
 {u^{3/2}\sqrt{4y-u}}\dd u
 \notag\\
 &=
 \frac{y}{\pi}\int_0^{4y}
 \frac{\sqrt u\,f_a(u)}{\sqrt{4y-u}}\dd u
 =
 (\mathcal A_{\mathrm{disk}}f_a)(y).
\end{align}
This proves \eqref{eq:diskAbelcovariance}.

At equal parameters, the common expansion
$\cK_X(u)=u^{-1}-1/2+o(1)$ gives
$F_0(u)=u-u^2/2+o(u^2)$, $F_0(0)=0$, and $F_0'(0)=1$.
The lower-end boundary term in \eqref{eq:diskIBP} is therefore
$-G_y(0)/\pi=-y$.
On the other hand, the diagonal self-correlation supported at the
coincident spectral coordinates $w=0$ contributes $G_y(0)/\pi=y$.
The two contributions cancel:
\begin{equation}
 \underbrace{-y}_{\text{lower boundary of the distinct-level part}}
 +
 \underbrace{y}_{\text{diagonal self-correlation at }w=0}
 =0.
 \label{eq:contactcancellation}
\end{equation}
The remaining integral in \eqref{eq:diskIBP} is
$(\mathcal A_{\mathrm{disk}}\cK_X)(y)$, proving
\eqref{eq:staticdiskAbel}.

Finally, the endpoint limits used above are integrable.  For $a>0$,
$G_y''(u)F_a(u)=O(u^{1/2})$ as $u\downarrow0$, while for $a=0$ it is
$O(u^{-1/2})$.  At the upper endpoint it is
$O((4y-u)^{-1/2})$.  Hence the integrations by parts follow by taking
$\varepsilon\downarrow0$.
\end{proof}

Equation \eqref{eq:contactcancellation} displays the role of the contact
term explicitly: the coincident-point self-correlation cancels the lower
boundary contribution of the distinct-level covariance, leaving precisely
the disk Abel transform \eqref{eq:staticdiskAbel}.
\subsection{Small windows}

For $q>-3/2$ the elementary substitution $u=4y\sin^2\theta$ gives
\begin{equation}
 \int_0^{4y}\frac{u^{q+1/2}}{\sqrt{4y-u}}\,du
 =(4y)^{q+1}B(q+\tfrac32,\tfrac12).
 \label{eq:Abelmoments}
\end{equation}
Here $B(x,y)=\Gamma(x)\Gamma(y)/\Gamma(x+y)$ is Euler's beta function.
Thus a term $k_qu^q$ in the kernel contributes
$yk_q(4y)^{q+1}B(q+3/2,1/2)/\pi$.  In particular, the contributions for
$q=-1,0,1,2$ are respectively $y$, $2k_0y^2$, $6k_1y^3$, and
$20k_2y^4$.  Differentiating \eqref{eq:Abelmoments} with respect to $q$
treats logarithmic terms.  We obtain
\begin{align}
 \Sigma_A^2(y)&=y-y^2+y^3-\frac56y^4+O(y^5),\nonumber\\
 \Sigma_{\AiD}^2(y)&=y-y^2-\frac{y^3}{2}
 [\log y+\gamma-\log2-\tfrac56]+o(y^3),\nonumber\\
 \Sigma_{\AiiD}^2(y)&=y-y^2+\frac23y^3-\frac{28}{45}y^5+O(y^6).
 \label{eq:smallwindows}
\end{align}
The symbol $\gamma$ in \eqref{eq:smallwindows} is Euler's constant and
should not be confused with the trace multiplier $\gamma_X$.

The linear and quadratic terms, $y-y^2$, are common to all three
symmetry classes; the first symmetry-dependent contribution appears at
cubic order, with a logarithmic modification in class $\AiD$.  This
structure can be anticipated directly.  The equal-eigenvalue contact
term contributes $y$, while $R_2^X(0)=0$ eliminates an order-$y^2$
contribution from distinct eigenvalue pairs, leaving $-y^2$ from the
square of the mean count.  Within a disk containing $y$ eigenvalues on
average, typical squared separations satisfy $u=O(y)$.  The behaviors
$R_2^X(u)=O(u)$ in classes $A$ and $\AiiD$, and
$R_2^{\AiD}(u)=O(u|\log u|)$, therefore produce corrections of order
$y^3$ and $y^3\log y$, respectively.  The precise coefficients are determined by the small-$u$ expansions of
the symmetry-class-dependent kernels in \eqref{eq:threeK}, through the
moment formula \eqref{eq:Abelmoments}.

The logarithmic term is a distinctive feature of class $\AiD$.  Indeed,
the small-$u$ expansion in \eqref{eq:threeK} gives
\[
 \cK_{\AiD}(u)
 =
 \frac1u-\frac12
 -\frac{u}{12}
 \bigl(\log u+\gamma-\log2-2\bigr)+o(u).
\]
Consequently, the static relation \eqref{eq:static} yields
\begin{align}
 \pi^2R_2^{\AiD}(u)
 &=1+\frac{\dd^2}{\dd u^2}
 \left[u^2\cK_{\AiD}(u)\right]
 \notag\\
 &=-\frac{u}{2}\log u+O(u),
 \qquad u\downarrow0.
 \label{eq:AIlogrepulsion}
\end{align}
Equivalently, in terms of the spectral distance $r=\sqrt u$,
$\pi^2R_2^{\AiD}(r^2)=-r^2\log r+O(r^2)$.
This logarithmically corrected quadratic pair repulsion was derived from the
Hermitian/non-Hermitian duality in \cite{ChenXiaoLiuRyu2026}; its origin in
the critical $k=1/2$ Calogero collision problem is explained by the exact
joint-density construction of \cite{XiaoChenLiuRyu2026}.
The two-dimensional radial Jacobian contributes one additional power of
$r$, yielding the associated logarithmically corrected cubic
nearest-neighbour behavior
$p_{\mathrm{NN}}(r)\propto-r^3\log r$.
The bulk and edge data of \cite{AkemannEtAl2026Edge} are compatible with
this law, although over the numerically accessible range they do not clearly
distinguish it from a pure cubic.

The kernel contains $-(u/12)\log u$.  Under the disk-overlap transform
\eqref{eq:numbervariance}, it contributes $-\tfrac12y^3\log y+O(y^3)$,
precisely the logarithmic term in
$\Sigma_{\AiD}^2(y)$.  Thus the unusual short-distance repulsion of
class $\AiD$ remains visible after integrating the point process over a
small disk. In class $\AiiD$, by contrast, the coefficient of $u^2$ in the kernel
vanishes, which accounts for the absence of a $y^4$ term in
\eqref{eq:smallwindows}.  

\subsection{Perimeter coefficients and the first correction}

Let $d_X=2/\beta$, where $\beta=2,1,4$ for $A,\AiD,\AiiD$.  The kernels
satisfy $\cK_X(u)=d_Xu^{-2}+O(u^{-3})$ at infinity.  Besides
\eqref{eq:kappa}, define the convergent subtracted moment
\begin{equation}
 \mathcal M_X=\int_0^\infty
 u^{3/2}\left[\cK_X(u)-\frac{d_X}{u^2}\right]du.
 \label{eq:subtractedmoment}
\end{equation}
The integral in \eqref{eq:subtractedmoment} is absolutely convergent.
Indeed, using the small-$u$ expansion
$\cK_X(u)=u^{-1}-1/2+o(1)$, the integrand is
$-d_Xu^{-1/2}+u^{1/2}+O(u^{3/2})$ as $u\downarrow0$.  Since
$\cK_X(u)=d_Xu^{-2}+O(u^{-3})$, it is $O(u^{-3/2})$ as $u\to\infty$.
Both endpoint behaviors are integrable.

\begin{proposition}[Large-disk expansion]
For the three specified kernels,
\begin{equation}
 \Sigma_X^2(y)=\kappa_X\sqrt y+\frac{\beta_X}{\sqrt y}
 +o(y^{-1/2}),\qquad \beta_X=\frac{\mathcal M_X}{16\pi}.
 \label{eq:perimeterexpansion}
\end{equation}
The leading coefficients are \eqref{eq:kappas}, and the corrections are
\begin{equation}
 \begin{aligned}
 \beta_A&=-\frac1{16\sqrt\pi},\\
 \beta_{\AiD}&=-\frac1{8\sqrt\pi}
 \left[1+\frac{3\pi}{2\sqrt2}-\frac{3\log(1+\sqrt2)}{\sqrt2}\right],\\
 \beta_{\AiiD}&=\frac{3\log(1+\sqrt2)-\sqrt2}{64\sqrt\pi}.
 \end{aligned}
 \label{eq:betas}
\end{equation}
In particular, there is no order-one term in any of the three classes.
\end{proposition}

\begin{proof}
Put $L=4y$.  After extracting the leading coefficient from
\eqref{eq:numbervariance}, the difference divided by $\sqrt y/(2\pi)$ is
\begin{equation}
 \int_0^L\sqrt u\,\cK_X(u)[(1-u/L)^{-1/2}-1]du
 -\int_L^\infty\sqrt u\,\cK_X(u)du.
 \label{eq:perimeterdifference}
\end{equation}
The contribution of $d_X/u^2$ is exactly zero, because
$\int_0^1x^{-3/2}[(1-x)^{-1/2}-1]dx=2=
\int_1^\infty x^{-3/2}dx$.
Replace the kernel in \eqref{eq:perimeterdifference} by
$R_X(u)=\cK_X(u)-d_X/u^2$.  On $[0,L/2]$,
$L[(1-u/L)^{-1/2}-1]\to u/2$ with a bound by a constant times $u$.
The dominating function $u^{3/2}|R_X(u)|$ is integrable.  On $[L/2,L]$ the
tail bound $R_X(u)=O(u^{-3})$ makes the contribution to $L$ times the integral
$O(L^{-1/2})$; the omitted tail beyond $L$ has the same bound.
Hence \eqref{eq:perimeterdifference} is
$\mathcal M_X/(2L)+o(L^{-1})$.  This proves
\eqref{eq:perimeterexpansion}, including the absence of a constant.
The evaluations of both moments are given in \cref{app:moments}.
\end{proof}

This proof also identifies the failure of a naive expansion under the
original Abel integral: its first unregularized moment diverges, and the
upper endpoint contributes at the same order as the subtracted tail.
For a complementary Mellin description, define
$M_X(s)=\int_0^\infty u^{s-1}\cK_X(u)\dd u$, initially for
$1<\Re s<2$.  Choose any vertical contour $\Re s=\sigma$ with
$1<\sigma<3/2$.  Mellin inversion of $\cK_X$ in the disk Abel transform
then gives
\begin{equation}
 \Sigma_X^2(y)=\frac1{4\pi}\frac1{2\pi i}
 \int_{\sigma-i\infty}^{\sigma+i\infty}
 M_X(s)B(\tfrac32-s,\tfrac12)(4y)^{2-s}\dd s.
 \label{eq:mellinabel}
\end{equation}
Conversely, taking the Mellin transform with respect to $y$ gives
\[
 M_X(s)=
 \frac{4\pi\,4^{s-2}}{B(\tfrac32-s,\tfrac12)}
 \int_0^\infty y^{s-3}\Sigma_X^2(y)\dd y,
 \qquad 1<\Re s<\frac32.
\]
Thus knowledge of $\Sigma_X^2(y)$ for all $y>0$ determines $M_X(s)$
throughout this strip.

The leading tail $\cK_X(u)\sim d_Xu^{-2}$ produces a simple pole of
$M_X(s)$ at $s=2$.  Continuing past this pole by subtracting that tail,
as in \eqref{eq:subtractedmoment}, gives
$M_X(5/2)=\mathcal M_X$.
The half-integer poles of the beta factor produce perimeter corrections;
its zeros at $s=2,3,\ldots$ cancel the simple poles from the algebraic kernel
tails.  Thus the absence of integer powers is common to the three kernels.
Endpoint subtraction and integration by parts in $\log u$ justify shifting
this contour across any fixed finite number of poles. The real-variable calculation above establishes these coefficients
directly, while the Mellin representation identifies their origin in the
pole structure of the transform.

\section{Direct matrix checks}
\label{sec:numerics}

We sampled the Gaussian matrices from the constructions in
\cref{eq:AIcov,eq:selfdual}, including their diagonal variances and the
self-dual projection.  For each class and $N=32,64,96$ we generated 3000
independent pairs $(X,Y)$.  At $z_0=0$ we used
$X_a=e^{-a/(2N)}X+\sqrt{1-e^{-a/N}}Y$ for $a=0.3,1$, and counted eigenvalues
in $|z|^2<y/N$, for $y=1,4$.  In the self-dual class each doublet contributed
one point.  The window area fixes the limiting mean count $y$; no finite-$N$
density fit or fitted parameter mass was used.  For numerical efficiency, we reuse the same realization of $Y$ at
different values of $a$.  For each fixed $a$ this gives the correct
two-matrix marginal distribution, while correlating the numerical
estimates obtained at different $a$; each column is therefore interpreted
as a separate two-time test.

Table~\ref{tab:counts} displays the results for $y=4$.  Error bars denote
one estimated Monte Carlo standard error, obtained from the sample
standard deviation of
$(\mathcal N_0-\overline{\mathcal N}_0)
(\mathcal N_a-\overline{\mathcal N}_a)$ divided by the square root of
the sample size.  Values at different $a$ obtained from the same matrix
pairs are correlated.

For $N=96$ in class $\AiD$, the preliminary batch of $3000$ pairs gave
at $a=1$ an estimate approximately three preliminary standard errors
above the kernel prediction.  This prompted a follow-up batch of $10000$
independent pairs, whose size was fixed in advance.  The displayed row
is computed from the pooled set of $13000$ raw observations; the two
batches are also reported separately in the accompanying numerical
data.  Since the follow-up calculation was prompted by the preliminary
deviation, its error bars are interpreted as descriptive Monte Carlo
uncertainties.  The table therefore provides a finite-$N$ consistency check of the kernel
predictions.
\begin{table}[htbp]
\centering
\small
\begin{tabular}{c r c c c}
\toprule
class & $N$ & $K_X(0,4)$ & $K_X(0.3,4)$ & $K_X(1,4)$\\
\midrule
$A$ & 32 & $1.0856\pm0.0273$ & $0.6765\pm0.0237$ & $0.4492\pm0.0219$\\
$A$ & 64 & $1.1313\pm0.0287$ & $0.6742\pm0.0242$ & $0.4452\pm0.0218$\\
$A$ & 96 & $1.0890\pm0.0281$ & $0.6362\pm0.0238$ & $0.4301\pm0.0231$\\
$A$ & $\infty$ & 1.1103 & 0.6615 & 0.4274\\
\midrule
$\AiD$ & 32 & $1.3673\pm0.0338$ & $0.8995\pm0.0285$ & $0.6332\pm0.0267$\\
$\AiD$ & 64 & $1.2906\pm0.0338$ & $0.8625\pm0.0288$ & $0.6190\pm0.0271$\\
$\AiD$ & 96 & $1.3447\pm0.0168$ & $0.9028\pm0.0144$ & $0.6426\pm0.0132$\\
$\AiD$ & $\infty$ & 1.3496 & 0.8893 & 0.6287\\
\midrule
$\AiiD$ & 32 & $0.9759\pm0.0247$ & $0.5563\pm0.0205$ & $0.3029\pm0.0190$\\
$\AiiD$ & 64 & $0.9606\pm0.0255$ & $0.5450\pm0.0209$ & $0.3062\pm0.0188$\\
$\AiiD$ & 96 & $0.9645\pm0.0236$ & $0.5080\pm0.0193$ & $0.2840\pm0.0183$\\
$\AiiD$ & $\infty$ & 0.9536 & 0.5119 & 0.2951\\
\bottomrule
\end{tabular}
\caption{Disk-count covariance from direct matrix sampling.  The
$N=\infty$ rows evaluate the kernel prediction, not an extrapolation of the
data.  The $N=96$, $\AiD$ row uses 13000 samples; the other finite-$N$ rows
use 3000.}
\label{tab:counts}
\end{table}

These data resolve the different class-dependent covariances and are
compatible with the proposed parameter normalization.  Sampling error and
finite-size effects remain visible.  The $y=1$ results and the
number-difference variances are included in the accompanying data as
additional checks of the same calculations.

For the Kramers law we used all distinct eigenvalues with $|z|<0.6$ in the
same self-dual matrices.  Numerically paired doublets were orthonormalized by a two-column QR
decomposition, $R=QT$, so that \eqref{eq:OK} becomes
$\cO_K=|q_1^T\mathbb Jq_2|^{-2}$, which avoids a basis-dependent vector
normalization, and we evaluated $m=\cO_K/[2N(1-|z|^2)]$ at each retained
point.
The inverse-gamma prediction gives
\[
 \E(m^{-1})=2,\qquad \E(m^{-2})=6,\qquad
 \Pr(m\le x)=e^{-1/x}(1+x^{-1}).
\]

Table~\ref{tab:Koverlaps} reports the first two inverse moments of the
rescaled Kramers-plane overlap $m$, together with the cumulative
probability $\Pr(m\leq1)$. To account for
correlations among eigenvalues of the same matrix, the error bars are
computed across the $J=3000$ independent matrix realizations, rather than
across the individual doublets.  Let $\mathcal R_j$ be the set of retained
Kramers doublets in the $j$th matrix, and define
$A_j=\sum_{i\in\mathcal R_j}h(m_i)$ and
$B_j=|\mathcal R_j|$.  Then
$\widehat\mu=\sum_jA_j/\sum_jB_j$, and its estimated standard error is
$\operatorname{sd}_j(A_j-\widehat\mu B_j)/(\overline B\sqrt J)$, where
$\overline B=J^{-1}\sum_jB_j$.

\begin{table}[htbp]
\centering
\small
\begin{tabular}{r r c c c}
\toprule
$N$ & retained doublets & $\E(m^{-1})$ & $\E(m^{-2})$
& $\Pr(m\leq1)$\\
\midrule
32 & 35474 & $2.0124\pm0.0104$ & $5.8780\pm0.0579$
& $0.7493\pm0.0032$\\
64 & 70061 & $2.0011\pm0.0080$ & $5.9160\pm0.0457$
& $0.7410\pm0.0023$\\
96 & 104590 & $2.0104\pm0.0067$ & $6.0116\pm0.0389$
& $0.7416\pm0.0020$\\
\midrule
$\infty$ & --- & $2$ & $6$ & $2/e=0.7358\ldots$\\
\bottomrule
\end{tabular}
\caption{Statistics of the rescaled Kramers-plane overlap $m$, obtained
from $3000$ independent matrices at each size.  The final row gives the
values implied by the conjectured limiting density
\eqref{eq:Kprediction}.}
\label{tab:Koverlaps}
\end{table}

The inverse moments are compatible with the predicted limit.  The
distribution-function column still shows finite-size differences: at
$N=96$ the measured $\Pr(m\le1/2)$ is $0.4119\pm0.0022$, compared with
$3e^{-2}=0.4060\ldots$.  The residual differences indicate visible finite-$N$ corrections.  The
data support the conjectured limiting distribution, while its first-moment
asymptotic additionally requires control of the rare large-overlap tail.

\section{Eigenvector overlaps and mesoscopic spectral dynamics}
\label{sec:overlaps}

\subsection{The exact \texorpdfstring{class-$A$}{class-A} overlap identity
and its conjectured extension}

In class $A$, set $m_i=O_{ii}/[N(1-|z_0|^2)]$.  Consider two distinct
eigenvalues at $z_1,z_2\to z_0$, with
$N|z_1-z_2|^2=u>0$, and let $m_1,m_2$ be their rescaled overlaps.
The class-$A$ pair intensity, conditional overlap moment, and
overlap-weighted pair intensity are
\begin{align}
 g_A(u)
 &:=\pi^2R_2^A(u)=1-e^{-u},
 \notag\\
 M_A(u)
 &:=\lim_{N\to\infty}\E(m_1m_2\mid z_1,z_2)
 =\frac{1+u^2-e^{-u}}{u^2(1-e^{-u})},
 \label{eq:BDconditional}\\
 W_A(u)
 &:=\pi^2R_2^A(u)M_A(u)
 =g_A(u)M_A(u)
 =1+\frac{1-e^{-u}}{u^2}
 =1+\cK_A(u).
 \label{eq:overlapA}
\end{align}
The formula for $M_A$ is the microscopic bulk limit of
\cite[Theorem~1.4]{BourgadeDubach2020}.  Thus, in class $A$, the memory
kernel $\cK_A=W_A-1$ is exactly the connected part of the
overlap-weighted pair intensity.

\begin{conjecture}[Overlap-weighted pair identity]
For $X\in\{\AiD,\AiiD\}$, associate with each distinct eigenvalue the
dimensionless overlap
$m_{\AiD,i}=2O_{ii}/(N\chi_0)$ or
$m_{\AiiD,i}=\cO_{K,i}/(2N\chi_0)$, respectively.  In class $\AiiD$,
each Kramers doublet is counted once.  In the same microscopic bulk limit,
define
\begin{equation}
 \begin{aligned}
 W_X(u)
 &:=
 \pi^2R_2^X(u)
 \lim_{N\to\infty}
 \E(m_{X,1}m_{X,2}\mid z_1,z_2),\\
 W_X(u)
 &=1+\cK_X(u),
 \qquad X\in\{\AiD,\AiiD\}.
 \end{aligned}
 \label{eq:overlapconjecture}
\end{equation}
\end{conjecture}

The characteristic-polynomial generating functions in
\cref{prop:HS} contain no eigenvector-overlap factors and therefore do
not establish \eqref{eq:overlapconjecture}.  A natural route to a direct
derivation is the Kac--Rice construction of joint
eigenvalue--eigenvector densities introduced in
\cite{FyodorovKacRice2025} and adapted to two correlated class-$A$
matrices in \cite{CupponeFyodorov2026}.  This construction keeps the
normalized eigenvectors explicit at finite $N$, allowing their overlap
factors to be included before the large-$N$ saddle-point analysis.  Its
extension to class $\AiD$ must incorporate the complex-symmetric
eigenvector constraint, while class $\AiiD$ requires treating the full
two-dimensional Kramers eigenspace and its basis-independent overlap
$\cO_K$.

\subsection{What is already rigorous at mesoscopic scales}

For non-Hermitian Ornstein--Uhlenbeck evolution in symmetry class $A$,
Bourgade \emph{et al.} \cite{BourgadeCipolloniHuang2024} prove a space--time central limit theorem
for smooth linear statistics, including nonequilibrium initial conditions
with independent identically distributed entries.  Their theorem does not
cover the transposition-symmetric classes $\AiD$ and $\AiiD$.  Their overlap
correlation consequence is an equilibrium statement.  At mesoscopic
separations it is summarized by
\begin{equation}
 \Cov(O_{ii}(t_1),O_{jj}(t_2))
 \sim\frac{\chi_0^2}
 {[\chi_0|t_1-t_2|+|\lambda_i(t_1)-\lambda_j(t_2)|^2]^2}.
 \label{eq:BCHschematic}
\end{equation}
This notation is distributional: their Corollary~2.11 pairs centered
empirical spectral measures weighted by diagonal eigenvector overlaps with
smooth spatial test functions and integrates over separated time intervals.
It is not a pointwise conditional assertion for
a prescribed pair of microscopic eigenvalues.  The admissible mesoscopic scale is $N^{-\eta}$ with $0<\eta<1/2$, the
corresponding times being of order $N^{-2\eta}$.

In the overlap region between our microscopic variables and those scales,
$u+a=N[|z-w|^2+\chi_0|t_1-t_2|]$.  The class-$A$ behavior
$\cK_A(u+a)\sim(u+a)^{-2}$ has the same quartic decay.  It supplies a
Gaussian microscopic candidate compatible with that established
mesoscopic law. For the transposition-symmetric classes, the replica calculation predicts
the large-separation tail
$\cK_X(u+a)\sim(2/\beta)(u+a)^{-2}$, with $\beta=1$ for $\AiD$ and
$\beta=4$ for $\AiiD$.  Interpreting the prefactor $2/\beta$ as the
coefficient of an overlap-weighted two-eigenvalue correlation additionally
requires \eqref{eq:overlapconjecture} or its two-time analogue.

These two levels of description are complementary.  In class $A$, the
rigorous mesoscopic limits for smooth linear statistics provide multitime
covariance information, but they do not resolve the joint microscopic
behavior of individual eigenvalues and their eigenvector overlaps.  The
replica calculation instead gives an explicit microscopic spectral
response for all three classes.  In the transposition-symmetric classes,
identifying this response with an overlap-weighted two-eigenvalue
correlation requires a finite-$N$ calculation containing explicit
eigenvector-overlap factors.  At mesoscopic scales, temporal memory in the coupled spectral and
eigenvector-overlap fluctuations is captured by the non-Markovian limiting
field identified in \cite{BourgadeCipolloniHuang2024}.

\section{Discussion}

The calculation relates the entrywise correlation of the Gaussian matrix
pair $(X_1,X_\alpha)$ to its microscopic two-time spectral response. 
 At finite $N$, this relation is encoded in the exact auxiliary-field
integral representations of \cref{prop:HS}. In the limit $N\to\infty$ with the positive integer $n$ fixed, the
covariance deformation and the microscopic spectral separation contribute
additively to the coefficient of the same invariant,
$\Tr(U\Lambda_XU^\dagger\Lambda_X)$, in the effective compact-group action.
The resulting integral therefore depends only on the single combination
$u+a$, with $\Lambda_X$ specified explicitly for each symmetry class.
Adopting the replica continuation used in the static
Hermitian/non-Hermitian duality then yields the conjectured two-time
density-correlation kernels.  Exact integration over the disk geometry
gives the cross-time eigenvalue-count covariances, equal-time number
variances, and large-window rigidity coefficients.

At short times, microscopic spectral diffusion is governed by the
amplification of matrix noise by eigenvector nonorthogonality.  Its
normalization is fixed independently by the known diagonal-overlap laws
in classes $A$ and $\AiD$, and in class $\AiiD$ by the exact velocity
formula involving the basis-independent Kramers-plane overlap.  
As shown in \eqref{eq:infrared}, at large parabolic separation the three
kernels satisfy $\cK_X(s)\sim2/(\beta_Xs^2)$, where $\beta_X=2,1,4$ is
the Dyson index of the corresponding Hermitian class.  Thus their leading
long-range coefficient retains the simple threefold dependence on
$\beta_X$, although the exact $\AiD$ and $\AiiD$ joint eigenvalue
densities are not pairwise Coulomb gases but contain genuine
many-eigenvalue interactions.

The main unresolved mathematical step for the two-time spectral limit is
control of the replica continuation near $n=0$, including the required
source derivatives.  Natural next steps are an independent derivation of the conjectured
Kramers-plane overlap law in class $\AiiD$, finite-$N$ derivations of the
overlap-weighted two-eigenvalue identities, universality beyond the
Gaussian ensembles, and the edge regime where the radial Hessian becomes
degenerate.  The explicit kernels, normalization identities, and matrix
simulations obtained here provide quantitative benchmarks for these
extensions, while the velocity calculation identifies the underlying
physical mechanism of overlap-amplified matrix noise.

The real Ginibre ensemble provides a different but closely related extension.
At a complex bulk point whose distance from the real axis remains fixed as
$N\to\infty$, the additional auxiliary-field modes induced by the reality
condition are expected to remain massive, leaving precisely the class-$A$
microscopic parametric kernel derived here.  When
$\operatorname{Im}z_0=O(N^{-1/2})$, however, these modes become soft and
should generate a new dynamical crossover involving the depleted complex
spectrum, the real eigenvalues, and the creation or annihilation of
complex-conjugate pairs.  Developing this real-Ginibre crossover requires a
separate analysis.

\section*{Acknowledgements}

The research reported in this paper was supported by UKRI grant UKRI1015,
``Non-Hermitian Random Matrices: Theory and Applications.''  

The author acknowledges assistance from ChatGPT (OpenAI) in improving the
organization and clarity of the presentation, performing consistency checks,
and designing numerical checks.  Claude Opus 5 (Anthropic) was used for numerical computations
and for critical reading.  All scientific arguments, results, and conclusions
remain the responsibility of the author.

\appendix

\section{Evaluation of the kernel moments}
\label{app:moments}

\subsection{Leading perimeter coefficients}

In class $A$,
\[
 \int_0^\infty u^{-3/2}(1-e^{-u})du=2\sqrt\pi,
\]
which gives $\kappa_A=1/\sqrt\pi$.  For the complex-symmetric class,
differentiating \eqref{eq:Adef} first gives the useful closed expression
\begin{equation}
 \mathfrak A(s)=\sinh\frac s2[E_1(s/2)-e^{-s/2}]
 -\frac s2\cosh\frac s2 E_1(s/2).
 \label{eq:Aclosed}
\end{equation}
Equivalently,
\[
 \cK_{\AiD}(s)=\cK_A(s)+
 \frac{[s\cosh(s/2)-2\sinh(s/2)]E_1(s/2)}{s^2}.
\]
The additional term is nonnegative.  Insert
$E_1(u/2)=\int_1^\infty e^{-wu/2}dw/w$ and integrate over $u$.
Writing $\ell=\log(1+\sqrt2)$ gives
\begin{align}
 \kappa_{\AiD}&=\frac1{\sqrt\pi}-\frac{J}{\pi},\nonumber\\
 J&=\sqrt{\frac\pi2}\int_1^\infty\frac{dw}{w}
 \left[\sqrt{w+1}-\sqrt{w-1}
 -\frac12\left(\frac1{\sqrt{w-1}}+\frac1{\sqrt{w+1}}\right)\right].
 \label{eq:kappaAIintegral}
\end{align}
The bracket is $O(w^{-5/2})$.  With $y_\pm=\sqrt{w\pm1}$, an
antiderivative of the bracket divided by $w$ is
\[
 2(y_+-y_-)+\frac12\log\frac{y_+-1}{y_++1}+\arctan y_-.
\]
Its values at $w=1$ and $w=\infty$ are $2\sqrt2-\ell$ and $\pi/2$.
Consequently $J=\sqrt{\pi/2}(\pi/2-2\sqrt2+\ell)$, which proves the first
line of \eqref{eq:kappas}.

For $\AiiD$, write $\operatorname{Shi}(u)=\int_0^1\sinh(ut)dt/t$.
The three contributions to
\[
 I=\int_0^\infty u^{-3/2}e^{-u}
 [(1+u)\operatorname{Shi}(u)+\sinh u]du
\]
are
\begin{align*}
 I_1&=\sqrt\pi\int_0^1\frac{\sqrt{1+t}-\sqrt{1-t}}t\,dt
 =\sqrt\pi(2\sqrt2-2\ell),\\
 I_2&=\frac{\sqrt\pi}{2}\int_0^1
 \frac{(1-t)^{-1/2}-(1+t)^{-1/2}}t\,dt=\sqrt\pi\ell,\\
 I_3&=\frac12\int_0^\infty u^{-3/2}(1-e^{-2u})du=\sqrt{2\pi}.
\end{align*}
The parameter integrals follow by $w=\sqrt{1\pm t}$; their logarithmic
endpoint terms cancel between the two branches.  Thus
$I=\sqrt\pi(3\sqrt2-\ell)$ and
$\kappa_{\AiiD}=I/(4\pi)$.  These calculations establish both the values and
the ordering $\kappa_{\AiD}>\kappa_A>\kappa_{\AiiD}$.

\subsection{Subtracted moments and the first perimeter correction}

The integral in \eqref{eq:subtractedmoment} is the analytic continuation of
$M_X(s)$ to $s=5/2$.  For example, splitting at $u=1$ gives
\[
 M_X(s)=\int_0^1u^{s-1}\cK_X(u)du
 +\int_1^\infty u^{s-1}[\cK_X(u)-d_X/u^2]du+\frac{d_X}{2-s},
\]
which immediately agrees with \eqref{eq:subtractedmoment} at $s=5/2$.
In class $A$, $M_A(s)=-\Gamma(s-2)$, so
$\mathcal M_A=-\sqrt\pi$.

For $\AiD$, use \eqref{eq:Aclosed} and the same $w$ representation of $E_1$.
At $s=5/2$ the continued expression is
\begin{equation}
 \mathcal M_{\AiD}=\sqrt\pi[-1+\sqrt2 J_*],
 \label{eq:MAIcontinued}
\end{equation}
where
\[
 J_*=\FP\int_1^\infty\frac{dw}{w}
 \left[\frac{(w-1)^{-3/2}+(w+1)^{-3/2}}2
 -(w-1)^{-1/2}+(w+1)^{-1/2}\right].
\]
Here $\FP$ denotes the finite part obtained by removing the explicit
inverse-square-root divergence at $w=1$; the preceding convergent
$u$-integral fixes this prescription uniquely.  The substitutions
$t=\sqrt{w-1}$ and $q=\sqrt{w+1}$ give the primitives
\[
 -\frac1t-3\arctan t,\qquad
 \frac1q+\frac32\log\frac{q-1}{q+1},
\]
respectively.  Their finite endpoint values yield
$J_*=-3\pi/2-1/\sqrt2+3\ell$.  Therefore
\[
 \mathcal M_{\AiD}
 =-2\sqrt\pi\left[1+\frac{3\pi}{2\sqrt2}
 -\frac{3\ell}{\sqrt2}\right].
\]

For $\AiiD$ the $\operatorname{Shi}$ representation gives, initially in the
convergence strip,
\begin{align}
 M_{\AiiD}(s)={}&\frac{\Gamma(s-2)}4
 \int_0^1\frac{(1-t)^{2-s}-(1+t)^{2-s}}t\,dt\nonumber\\
 &+\frac{\Gamma(s-1)}4
 \int_0^1\frac{(1-t)^{1-s}-(1+t)^{1-s}}t\,dt
 -\frac{2^{2-s}}4\Gamma(s-2).
 \label{eq:MAIIcontinued}
\end{align}
At $s=5/2$, the first parameter integral is $2\ell$.  The finite part of
the second is $2\ell-\sqrt2$: subtract $2/\sqrt\varepsilon$ after cutting
the upper endpoint at $t=1-\varepsilon$, and use
$2/w+\log|(w-1)/(w+1)|$ with $w=\sqrt{1\pm t}$.
It follows that
\[
 \mathcal M_{\AiiD}=\frac{\sqrt\pi}{4}(3\ell-\sqrt2).
\]
Dividing the three subtracted moments by $16\pi$ proves
\eqref{eq:betas}.  Numerical values are
$\beta_A=-0.03526184897\ldots$,
$\beta_{\AiD}=-0.17366365293\ldots$, and
$\beta_{\AiiD}=0.01084220047\ldots$.
Direct integration of \eqref{eq:subtractedmoment}, using $u=t^2$ and the
algebraic tails beyond a large cutoff, reproduces these values.

\end{document}